%% file: main.tex
\documentclass[reqno,11pt,oneside]{amsart}
\usepackage[margin=2.5cm]{geometry}
\usepackage{amsmath}
\usepackage{mathrsfs} 
\usepackage{amssymb}
\usepackage{graphicx}
\usepackage{hyperref}
\usepackage{xspace}
\usepackage{bm}
\usepackage{slashed}
\usepackage{dsfont}
\usepackage{svg}
\usepackage{adjustbox}
\usepackage{ccaption}
\usepackage{subcaption}
\usepackage{csvsimple}
\usepackage{longtable}
\usepackage[shortlabels]{enumitem}
\usepackage{ifsym}

\usepackage{float}

\usepackage{xcolor}
\definecolor{cbGreen}{HTML}{4DAC26}
\definecolor{cbRed}{HTML}{D01C8B}
\usepackage{booktabs}

\usepackage{algorithm}
\usepackage{bbm}
\usepackage{algpseudocode}

\usepackage{tikz,lipsum,lmodern}
\usepackage[most]{tcolorbox}

\usepackage{natbib}
\usepackage{url}
\usepackage{mathtools}
\usepackage{multirow}
\graphicspath{{./figs/}}

\definecolor{light-gray}{gray}{0.8}

\newcommand{\Tr}{\mbox{Tr}}

\newcommand{\EDaff}{Higgs Centre for Theoretical Physics, School of Physics \& Astronomy, University of Edinburgh, Edinburgh EH9 3FD, United Kingdom.}
\newcommand{\IGCaff}{MRC Human Genetics Unit, Institute of Genetics \& Cancer, University of Edinburgh, Edinburgh EH4 2XU, United Kingdom.}
\newcommand{\Berkeleyaff}{Division of Biostatistics, University of
California, Berkeley, CA, USA}
\newcommand{\EDMathaff}{School of Mathematics and Maxwell Institute, University of Edinburgh, Edinburgh EH9 3FD, United Kingdom.}
\newcommand{\Informaticsaff}{School of Informatics, University of Edinburgh, Edinburgh EH8 9AB, United Kingdom} 
\newcommand{\EqualCorresponding}{Equal contribution, corresponding authors (alphabetical order)} 

\newcommand{\beginsupplement}{%
        \setcounter{table}{0}
        \renewcommand{\thetable}{S\arabic{table}}%
        \setcounter{figure}{0}
        \renewcommand{\thefigure}{S\arabic{figure}}%
     }

\newcommand{\indep}{\perp \!\!\! \perp}

\input{macros}

\usepackage{todonotes}

\begin{document}

%\title[TarGene]{Semi-parametric efficient estimation of n-point interactions with applications to population genetics}
\title[TarGene]{Semi-parametric efficient estimation of small genetic effects in large-scale population cohorts}
\author{
Olivier Labayle\textsuperscript{1,2},
Breeshey Roskams-Hieter\textsuperscript{1,2},
Joshua Slaughter\textsuperscript{1,2},
Kelsey Tetley-Campbell\textsuperscript{1,2},
Mark J. van der Laan\textsuperscript{3},
Chris P. Ponting\textsuperscript{1},
Sjoerd Viktor Beentjes\textsuperscript{1,3,4,*},
Ava Khamseh\textsuperscript{1,2,3,*}
}

\address{$^1$\IGCaff}
\address{$^2$\Informaticsaff}
\address{$^3$\Berkeleyaff}
\address{$^4$\EDMathaff}
\address{$^5$\EDaff}
%\address{$^*$\EqualContrib}
\address{$^{*}$\EqualCorresponding}

\begin{abstract}
Population genetics seeks to quantify DNA variant associations with traits or diseases, as well as interactions among variants and with environmental factors.
Computing millions of estimates in large cohorts in which small effect sizes are expected, necessitates minimising model-misspecification bias to control false discoveries.
We present TarGene, a unified statistical workflow for the semi-parametric efficient and double robust estimation of genetic effects including $k$-point interactions among categorical variables in the presence of confounding and weak population dependence.
$k$-point interactions, or Average Interaction Effects (AIEs), are a direct generalisation of the usual average treatment effect (ATE).
We estimate AIEs with cross-validated and/or weighted versions of Targeted Minimum Loss-based Estimators (TMLE) and One-Step Estimators (OSE).
The effect of dependence among data units on variance estimates is corrected by using sieve plateau variance estimators based on genetic relatedness across the units.
We present extensive realistic simulations to demonstrate power, coverage, and control of type I error.
Our motivating application is the targeted estimation of genetic effects on trait, including two-point and higher-order gene-gene and gene-environment interactions, in large-scale genomic databases such as UK Biobank and All of Us.
All cross-validated and/or weighted TMLE and OSE for the AIE $k$-point interaction, as well as ATEs, conditional ATEs and functions thereof, are implemented in the general purpose \texttt{Julia} package TMLE.jl.
For high-throughput applications in population genomics, we provide the open-source \texttt{Nextflow} pipeline and software TarGene which integrates seamlessly with modern high-performance and cloud computing platforms.
\vspace{0.5cm}
\end{abstract}

%\pacs{}

\maketitle

\setcounter{tocdepth}{1}

%\newpage
%\tableofcontents
%\newpage

\vspace{-0.8cm}

%%%%%%%%%%%%%%%%%%%%%%%%%%%%
%%%%%%% Introduction %%%%%%%
%%%%%%%%%%%%%%%%%%%%%%%%%%%%

\section{Introduction}\label{sec:Introduction}

Over the past 15 years, many genetic risk loci have been associated with disease or traits~\citep{Visscher2017}.
Most genetic effect sizes are small and require large samples to be reliably estimated~\citep{Uffelmann2021}.
Their detection helps guide interventions on the downstream disease consequences of the variant that may be both sizeable and clinically relevant~\citep{minikel2024,King2019}.
Nevertheless, a slight bias, such as when an estimation model is misspecified, can appreciably alter genetic effect estimates and, additionally, their associated confidence intervals may not have the theoretical nominal coverage of the ground truth.
This problem of bias is exacerbated for larger sample sizes that yield more confident estimates.
Statistical models in common use, such as logistic, linear and linear mixed models (LMMs)~\citep{Uffelmann2021}, rely on parametric assumptions, including linearity of the variant-covariate-trait relationship, and normality of the associated conditional distribution.
LMMs also usually include a random effect to capture population relatedness~\citep{Uffelmann2021,GRAMMAR-Gamma,fastlmm}.\\ 

When genotype assignment is entirely random (\ie, without population stratification), and when non-linear genotype effects are absent, the commonly employed LMM model is adequate for estimating the true effect size consistently.
Nevertheless, genetic effects are often non-linear~\citep{Neale-nonlinearity-science,Milind2024.11.11.24317065,Mackay2024}, and have been accounted for as an additive-plus-dominance part of an LMM~\citep{10.1371/journal.pgen.1006421}.
Also, warped LMM has generalised LMMs by transforming outcomes to account for non-Gaussian distributed residuals~\citep{Fusi2014}.
Recent applications of ML algorithms, such as DeepNull/XGBoost~\citep{deepnull2022}, can adaptively account for genetics and confounder/covariate non-linearities data-adaptively, yet they lack theoretical guarantees of consistency and coverage.
Furthermore, for any variant and trait combination, the validity of any linearity assumption is a priori unknown, and so provides no guarantee that the LMM is a consistent estimator with nominal coverage. \\

Beyond single variant associations, interactions among multiple variants, known as epistasis, may be even smaller and more complicated to estimate~\citep{Haley_Epistasis} yet are crucial to understanding complex disease~\citep{Mackay2024}.
In summary, robust estimation techniques are needed to accurately estimate genetic effect sizes, together with interactions among variants, and interactions among variants and environmental factors.\\

The Targeted Learning (TL) Roadmap offers a step-by-step guide to estimating causal and statistical parameters from real-world data with minimal bias \citep{Gruber2024}.
This includes the construction of substitution estimators, including Targeted Maximum Likelihood Estimators (TMLE), which are semi-parametric efficient under minimal assumptions~\citep{vanderLaanRubin+2006}.
TL has been successfully applied in various clinical and biomedical settings to generate real-world evidence \citep{Fong2022-TL,Gilbert2022-TL,Havlir2019-TL}. \\

Here, we introduce Targeted Genomic Estimation (TarGene), a method based on targeted semi-parametric estimation theory \citep{OS_Pfanzagl,MR2867111} following the TL roadmap for the estimation of genetic effects of single variants and interactions. 
The five key contributions of this paper are as follows. 
(1) Development of semi-parametric estimators, namely, Targeted Maximum Likelihood Estimator (TMLE), weighted TMLE (wTMLE), one-step estimator (OSE) and their cross-validated (CV) versions, for $k$-point interaction effects on outcome, here among genetic variants, and genetic variants and environmental factors. 
(2) Accounting for population dependence structure, by updating variance estimates via a network adaptation of the sieve plateau variance estimator of~\cite{SievePlateau}.
(3) Extensive simulations to assess type I error, power, and coverage of various semi-parametric estimators and target parameters in realistic population genetics contexts. This includes (rare) binary outcomes, continuous outcomes, and categorical variants with diverse population frequencies.
(4) A general purpose \texttt{Julia} software package, \href{https://targene.github.io/TMLE.jl/stable/}{TMLE.jl}, implements these semi-parametric estimators and $k$-point interaction target parameters.
(5) An end-to-end \texttt{Nextflow} pipeline, software and documentation, \href{https://targene.github.io/targene-pipeline/stable/}{TarGene}, for scalable and seamless application to large-scale biobanks ($0.5$ million or more individuals) such as the UK Biobank~\citep{UKB_bycroft} and the All of Us cohort~\citep{AllOfUs}. \\

As an illustration, we apply TarGene to five applications on UK Biobank data: (i) Performing a Phenome-wide Association Study (PheWAS) for the \textit{FTO} intronic variant rs1421085, a candidate causal variant for obesity~\citep{claussnitzerFTOObesityVariant2015a} on 768 traits, including body mass index (BMI).
We find that p-values currently reported in the literature may be inflated thereby leading to a higher false discovery rate.
(ii) We simultaneously discover non-linear effect sizes of additional allelic copies on trait or disease. Specifically, we demonstrate significant genetic non-linearity at the \textit{FTO} locus for 39 traits in this study.
(iii) For the same locus, we find evidence of gene-by-environment interactions (G$\times$E) using deprivation indices on body weight-related traits.
(iv) We further reproduce five pairs of epistatic loci associated with red hair and find 16 further epistatic interactions on hair or skin colour (G$\times$G).
(v) Finally, we illustrate how TarGene can be used to investigate higher-order interactions (G$\times$G$\times$G) using three variants linked to the vitamin D receptor complex.
These analyses exemplify the breadth of questions that can be addressed with TarGene through its three main \texttt{Nextflow} workflows: PheWAS, GWAS and custom single or joint variant effects, gene-by-gene interactions and gene-by-environment interactions up to any order. \\

The estimators developed and implemented in this work take into account population stratification, via genetic ancestry principal component analysis (PCA), as a confounder of the variant-to-outcome relationship, as well as covariates, such as age and sex, which can influence outcome.
However, in general, this does not imply that the estimated genetic effect sizes reflect causal mechanisms.
This is because of the co-inheritance of dependent variants, known as linkage disequilibrium (LD) blocks in the genome.
While TarGene estimators address any statistical gap due to model misspecification, the causal gap currently remains. 
In the statistics literature, attempts have been made to address the causal gap in population genetics analyses through fine-mapping, \eg, SuSiE~\citep{10.1111/rssb.12388,10.1371/journal.pgen.1010299}, and KnockOffGWAS~\citep{Sesiae2105841118}, a variable-selection method based on the knockoff filter of~\cite{10.1214/15-AOS1337}.
However, SuSiE and other fine-mapping methods such as FINEMAP~\citep{10.1093/bioinformatics/btw018} and genome-wide fine-mapping (GWFM)~\citep{Wu2024.07.18.24310667} rely on modelling assumptions and thus may not close the statistical gap.
While KnockoffGWAS makes no parametric assumptions regarding the distribution of the phenotype conditional on genotypes and controls the false discovery rate (FDR) whilst accounting for population structure, it computes neither effect sizes nor epistatic interactions, and only reports statistical significance.
These quantities are essential for explaining how variants, via biological mechanisms and regulatory functions, modify a trait or disease risk.
A unified method that closes both causal and statistical gaps in genomic medicine has yet to be developed.\\

This paper is organised as follows.
In Section~\ref{sec:Methods}, we first introduce the k-point interaction target parameter, derive its efficient influence function, and exact second-order remainder.
This allows for the construction of various semi-parametric efficient estimators and their cross-validated analogues.
We discuss various hypothesis testing strategies, and adapt sieve variance plateau estimators to account for population dependence in a network setting.
In Section~\ref{sec:Simulations}, we describe two extensive simulations: A null simulation of no genetic effect and a realistic simulation modelled on real UK Biobank data.
On these, we assess our estimators’ performance in terms of power and coverage.
Subsequently, we apply TarGene to perform three common population genetics analyses on UK Biobank data, namely a PheWAS, a gene-by-environment interaction study, and a third order interaction study in Section~\ref{sec:UKB}.
We also compare the genetic effect of the FTO variant on BMI between UKB and the All of Us cohort.
The final Section~\ref{sec:Software} contains details of our \texttt{Julia} package \href{https://targene.github.io/TMLE.jl/stable/}{TMLE.jl} and \texttt{Nextflow} pipeline \href{https://targene.github.io/targene-pipeline/stable/}{TarGene}, as well as runtime considerations.

%%%%%%%%%%%%%%%%%%%%%%%
%%%%%%% Methods %%%%%%%
%%%%%%%%%%%%%%%%%%%%%%%

%\section{Methods}\label{sec:Methods}

\newpage

\section{Efficient estimation of $n$-point interactions}\label{sec:Methods}
\subsection{Inferential problem}

Suppose that the observed data unit is
\begin{equation}
    O = (W,A_1,A_2,\ldots,A_m,Y) \sim P_0,
\end{equation}
where $W$ is a vector of pre-treatment covariates, the $A_i \in \{0,1,\ldots t_i\}$ with $i = 1, \ldots, m$ are categorical treatment variables, $Y$ is an outcome of interest, and $P_0 \in \CM$ is the true data-generating probability distribution.
In genomics, $Y$ denotes a disease or trait, the $A_i$ encode DNA variants at particular locations in the genome, and $W$ is a vector of genetic confounders due to ancestry and relatedness, usually captured by principal components.
Since we do not wish to impose any potentially unrealistic assumptions on the statistical model $\CM$ -- and hence $P_0$ -- we take $\CM = \CM_0$, the non-parametric model.
However, our derivations carry through under arbitrary restrictions on the distributions of the $A_i$ given $W$.
Note that in our applications to the UK Biobank in Section~\ref{sec:UKB}, the selected treatment variables are independent given the sources of population stratification $W$. \\

Our statistical estimand of interest is the $k$-point interaction among a subset of the treatment variables $A_1, \ldots, A_m$ in their effect on outcome $Y$ whilst correcting for covariates $W$.
This parameter is a generalised difference of mean outcomes, and best illustrated by an example.
Let $Q_{a_1,a_2}(w) = Q(a_1,a_2,w) = Q_P(a_1,a_2,w) = \BE_P(Y|A_1=a_1,A_2=a_2,W=w)$ denote the outcome regression, and let $Q_W(w) = P(W\leq w)$ denote the cumulative distribution function of the covariates $W$.
The $2$-point interaction parameter $\Psi_a^{(2)} \colon \CM_0 \to \BR$ of $A_1$ and $A_2$ in their effect on $Y$ is defined as
\begin{align}\label{eq:2pt}
    \Psi_a^{(2)}(P) 
    &= \int \Bigl[\bigl\{Q_{1,1}(w)-Q_{0,1}(w)\bigr\}-\bigl\{Q_{1,0}(w)-Q_{0,0}(w)\bigr\} \Bigr]dQ_W(w) \\
    &= \int \Bigl[\bigl\{Q_{1,1}(w)-Q_{1,0}(w)\bigr\}-\bigl\{Q_{0,1}(w)-Q_{0,0}(w)\bigr\} \Bigr]dQ_W(w),
\end{align}
where $P \in \CM_0$ and $a = (A_1 \colon 0 \to 1, A_2 \colon 0 \to 1)$ encodes the change in treatment levels for $A_1, A_2$.
The parameter $\Psi_a^{(2)}(P)$ quantifies the difference between the treatment effect of $A_1 \colon 0 \to 1$ on $Y$ given $A_2 = 1$ and the treatment effect of $A_1 \colon 0 \to 1$ on $Y$ given $A_2 = 0$, correcting for $W$ or, symmetrically, the difference between the treatment effect of $A_2 \colon 0 \to 1$ on $Y$ given $A_1 = 1$ and the treatment effect of $A_2 \colon 0 \to 1$ on $Y$ given $A_1 = 0$, correcting for the marginal distribution of $W$.
This parameter admits the following alternative presentation:
\begin{equation}\label{eq:2pt_alternative}
    \Psi_a^{(2)}(P) = \int \Bigl[\bigl\{Q_{1,1}(w)-Q_{0,0}(w)\bigr\}-\Bigl(\bigl\{Q_{1,0}(w)-Q_{0,0}(w)\bigr\} + \bigl\{Q_{0,1}(w)-Q_{0,0}(w)\bigr\}\Bigr) \Bigr]dQ_W(w).
\end{equation}
Thus, $\Psi_a^{(2)}(P)$ also represents the non-additive effect on $Y$ due to the joint change in the treatment levels of $A_1$ and $A_2$ simultaneously, rather than individually~\citep{hernan2023causal}.
Two-point interactions have previously been studied in the experimental design and causal inference literature~\citep{Cox1984Interaction,VanderWeele_interactions,Dasgupta2015CausalIF,EgamiImai2019}.\\

For example, let $A_1,A_2$ be binary treatment variables and consider a linear model with interaction term in which the coefficients can depend on the covariates $W$, namely 
\begin{equation}
    Y = \alpha(W) + \beta_1(W)A_1+\beta_2(W)A_2+\gamma(W) A_1A_2 + \epsilon, \qquad \epsilon \sim N(0,\sigma^2),
\end{equation}
and $\BE[\epsilon|W] = 0$.
The conditional mean $\BE(Y|A_1,A_2,W)$ is the same expression without the noise term $\epsilon$.
Evaluating the above target parameter of Eq.~\ref{eq:2pt_alternative} yields
\begin{equation}
    \Psi_a^{(2)}(P) = \int\Bigl\{ \bigl(\beta_1(W)+\beta_2(W)+\gamma(W)\bigr) - \bigl( \beta_1(W)+\beta_2(W)\bigr)\Bigr\} dQ_W(w) = \BE_P\bigl\{\gamma(W)\bigr\}.
\end{equation}
When $\gamma(W) = \gamma$ is independent of covariates, this reduces to $\Psi_a^{(2)}(P) = \gamma$.
It follows that the definition in Eq.~\ref{eq:2pt} coincides with the usual notion of $2$-point interaction in a linear model. \\

Next, we give the general definition of the $k$-point interaction among a subset of $m$ treatment variables in their effect on outcome $Y$ whilst correcting for covariates $W$.
Given $k$ of the $m$ treatment variables $A = (A_{j_1},\ldots,A_{j_k})$ all pairwise distinct, specify an initial treatment level $a(0) = (a_{j_1}(0),\ldots,a_{j_k}(0))$ of the $k$ variables and a final treatment level $a(1) = (a_{j_1}(1),\ldots,a_{j_k}(1))$ where $a_{j_i}(l) \in \{0,1,\ldots,t_{j_i}\}$ for each of the $k$ treatment variables $j_i$.
The $k$-point interaction generalises Eq.~\ref{eq:2pt_alternative} and quantifies the non-additive effect due to the joint change in treatment levels of the variables $A$ from $a(0)$ to $a(1)$ relative to the sum of the effects of marginal changes for all subsets of variables with the remaining treatment variables held at their initial level.
Notationally, it is more convenient to consider the generalisation of Eq.~\ref{eq:2pt}, defined in ~\cite{PhysRevE.102.053314}:
\begin{equation}\label{eq:k-point}
    \Psi^{(k)}_{a(0),a(1)}(P) = \sum_{s \in \{0,1\}^k} (-1)^{k - (s_1+\ldots+s_k)} \Psi_{a(s)}(P),
\end{equation}
where we abbreviate the levels of the $k$ treatment variables to $a(s) = (a_{j_1}(s_1),\ldots, a_{j_k}(s_k))$ and we denote the treatment-specific, covariate-adjusted mean outcome implied by $P \in \CM_0$ by
\begin{align}
    \Psi_{a(s)}(P) 
    &= \BE_P[\BE_P(Y|A = a(s),W)] \\
    &\equiv \BE_P[\BE_P(Y|A_{j_1}=a_{j_1}(s_1), \ldots, A_{j_k}=a_{j_k}(s_k),W)].
\end{align}
Under additional causal assumptions $\Psi_{a(s)}(P)$ can be interpreted as the causal mean counterfactual outcome $\BE[Y(a(s))]$ under the joint treatment assignment $A = a(s)$ where $s \in \{0,1\}^k$~\citep{Rubin1974}.\\

Before proceeding with the analysis of this target parameter, we recall some notation.
As the parameter depends only on $P$ through $Q = Q(P) = (\bar{Q},Q_W)$, we sometimes write $\Psi(Q)$ instead of $\Psi(P)$.
We denote the true outcome regression of $Y$ on $(A,W)$ by $\bar{Q}_0$, the true covariate distribution by $Q_{W,0}$, and abbreviate both by $Q(P_0) = Q_0 = (\bar{Q}_0, Q_{W,0})$.
We write
\begin{equation}
    g(w) = g_P(w) = P(A = a(s) \mid W = w) = P(A_{j_1} = a_{j_1}(s_1), \ldots, A_{j_k} = a_{j_k}(s_k) \mid W=w)
\end{equation}
for the propensity score, where $s \in \{0,1\}^k$ is a binary vector of length $k$.
Throughout, we assume that the true propensity score $g_0$ satisfies the positivity condition $\delta < g_0(w) < 1-\delta$ for some $\delta > 0$, all covariates $w$ in the support of $Q_{W,0}$, and all treatment assignments.
Given a probability distribution $P \in \CM_0$ and any $P$-integrable function $f$, we write
\begin{equation}
    Pf = \BE_{P}[f(O)] = \int f(o) dP(o).
\end{equation}
We denote the empirical distribution on $n$ variables $O_1,\ldots,O_n$ by $\BP_n$.
In particular, 
\begin{equation}
    \BP_n f = \frac{1}{n} \sum_{i=1}^n f(O_i)
\end{equation}
denotes the sample average of $f$ with respect to $\BP_n$.

\subsection{Influence function and exact remainder}

The efficient influence function is a key object in semi-parametric efficient estimation, as it characterises the asymptotic behaviour of any regular and efficient estimator~\citep{bickel1998efficient}.
After recalling the necessary ingredients, we derive the efficient influence function (EIC) of our $k$-point interaction estimator of Eq.~\ref{eq:k-point} and analyse its robustness properties.
We use it to construct non-parametric efficient estimators of interaction via (i) one-step bias correction~\citep{PfanzaglWefelmeyer+1985+379+388}, and (ii) targeted minimum loss-based estimation (TMLE)~\citep{vanderLaanRubin+2006,MR2867111,MR3791826}.\\

Recall that a regular estimator $\hat{\psi}_n$ of $\psi_0 = \Psi(Q_0)$ is asymptotically linear if it can be written as $\hat{\psi}_n = \psi_0 + \BP_n D(P_0) + o_P(n^{-1/2})$, where $D(P_0)$ is a gradient of the statistical estimand $\Psi$ at $P_0$ with respect to a statistical model $\CM$.
The function $D(P)(O)$ is a gradient of $\Psi$ at $P$ relative to $\CM$ if
\begin{equation}
    \frac{d}{d\epsilon} \Psi(P_{\epsilon}) \Big|_{\epsilon = 0} = \int D(P)(o) s(o) dP(o),
\end{equation}
for any regular one-dimensional parametric submodel $\{P_{\epsilon}\} \subseteq \CM$ such that $P_{\epsilon = 0} = P$ with score function $s(o)$ at $\epsilon = 0$, and if $D(P) \in L^2_0(P)$.
Here, for each $P \in \mathcal{M}$, we denote by $L^2_0(P)$ the Hilbert space of real-valued functions of $O$ with zero mean, finite variance, defined on the support of $P$, and endowed with the covariance inner product given by $\langle f,g \rangle = \BE_{P}[f(O)g(O)]$ for $f,g \in L^2_0(P)$.
We denote the induced norm by $|\!|-|\!|_P$.
The tangent space $\CT_{\CM}(P) \subseteq L^2_0(P)$ of $\CM$ at $P$ is the $L^2_0(P)$-closure generated by scores of regular one-dimensional parametric models.
For the nonparametric model $\CM_0$ we have $\CT_{\CM}(P) = L^2_0(P)$ for every $P \in \CM_0$.
There exists a unique canonical gradient $D^*(P) \in \CT_{\CM}(P)$, which is referred to as the efficient influence function as its squared $L^2(P)$-norm is the generalized Cramer--Rao (CR) lower bound for estimating $\Psi(P)$ relative to $\CM$~\citep{bickel1998efficient}.
Under sampling from $P_0$, a regular asymptotically linear (RAL) estimator is efficient if and only if its influence function is the efficient influence function.
The efficient influence function of $\Psi_a(P)$ at $P$ relative to $\CM_0$ is
\begin{equation}\label{eq:gradient_point_treatment}
    D_a^*(P)(O) = D_a^*(Q,g)(O) = \frac{\ID\{A=a\}}{g(A,W)}\Bigl\{ Y-\bar{Q}(A,W) \Bigr\} + \bar{Q}(a,W) - \Psi_a(Q),
\end{equation}
where $O$ is distributed according to $ P \in \CM_0$~\citep{van2003unified}.\\

Next, we derive the efficient influence function of the $k$-point interaction parameter $\Psi^{(k)}_{a(0),a(1)}(P)$ of Eq.~\ref{eq:k-point} as introduced in~\citep{PhysRevE.102.053314}.
Since the gradient of a linear combination of target parameters is equal to the linear combination of their gradients, we find:
\begin{lemma}\label{lem:IC_k_point}
    The efficient influence function of the  $k$-point interaction target parameter $\Psi^{(k)}_{a(0),a(1)}(P)$ of Eq.~\ref{eq:k-point} at $P$ relative to the non-parametric model $\CM_0$ is equal to
    \begin{equation}\label{eq:IC_k_point}
        D^{*}_{a(0),a(1)}(P) = \sum_{s \in \{0,1\}^k} (-1)^{k - (s_1+\ldots+s_k)} D_{a(s)}^*(P)
    \end{equation}
    where the expression of $D^*_{a(s)}(P)$ is given in Eq.~\ref{eq:gradient_point_treatment}.
\end{lemma}

The gradient $D^*(P_0)$ provides a first-order approximation of the parameter $\Psi(P)$ around the true value $\Psi(P_0)$ at the data-generating distribution $P_0$.
More precisely, the Von Mises expansion reads
\begin{equation}\label{eq:von_mises}
    \Psi(P) - \Psi(P_0) = \bigl(P-P_0) D^*(P) + R(P,P_0)
\end{equation}
where this equation defines the second-order exact remainder $R(P,P_0) = \Psi(P) - \Psi(P_0) + P_0 D^*(P)$.
The second-order exact remainder of the parameter $\Psi_a(P)$ at $P$ relative to $\CM_0$ can be written as
\begin{equation}\label{eq:exact_remainder_point_treatment}
    R_a(P,P_0) = R_a\bigl(Q_a,Q_{a,0},g_a,g_{a,0}\bigr) = P_0\Bigl\{ \bigl(\bar{Q}_a-\bar{Q}_{a,0}\bigr) \bigl(g_a-g_{a,0}\bigr)/g_a\Bigr\}
\end{equation}
where we write $\bar{Q}_a(W) = \bar{Q}(a,W)$, $g_a(W) = P(A=a \mid W)$ and their analogues implied by $P_0$.
In particular, the remainder of this parameter is double robust by the Cauchy--Schwarz inequality:
\begin{equation}\label{eq:point_treatment_is_double_robust}
    R_a(P,P_0) \leq |\!|\bar{Q}_{a,0}-\bar{Q}_a|\!|_{P_0} \cdot |\!|(g_{a,0}-g_a)/g_a|\!|_{P_0}.
\end{equation}
Note that the exact second-order remainder is linear in the target parameter because the gradient is linear in the target parameter, and hence so is its defining equation in Eq.~\ref{eq:von_mises}.
We immediately obtain the second-order exact remainder of our $k$-point interaction parameter of Eq~\ref{eq:k-point}.
\begin{lemma}\label{lem:exact_remainder_k_point}
    The second-order exact remainder of the  $k$-point interaction target parameter $\Psi^{(k)}_{a(0),a(1)}(P)$ of Eq.~\ref{eq:k-point} at $P$ relative to the non-parametric model $\CM_0$ is equal to
    \begin{align}\label{eq:exact_remainder_k_point}
        R_{a(0),a(1)}(P,P_0) 
        = \sum_{s \in \{0,1\}^k} (-1)^{k - (s_1+\ldots+s_k)} R_{a(s)}(P,P_0)
    \end{align}
    where the expression of $R_{a(s)}(P,P_0)$ is given in Eq.~\ref{eq:exact_remainder_point_treatment}.
\end{lemma}
Henceforth, we leave the interaction order, as well as the initial and final treatment levels, implicit and write $\Psi(P)$ for the $k$-order interaction from levels $a(0)$ to $a(1)$ as defined in Eq.~\ref{eq:k-point}.
Let $\hat{P}_n$ be an estimator of $P_0$ based on the available data $\BP_n$, and denote by $\hat{\psi}_n \equiv \Psi(\hat{P}_n)$ the corresponding plug-in estimator.
Since $\Psi$ is path-wise differentiable, we inspect its Von Mises expansion:
\begin{equation}
    \Psi(\hat{P}_n) - \Psi(P_0) = \bigl(\hat{P}_n-P_0) D^*(\hat{P}_n) + R(\hat{P}_n,P_0).
\end{equation}
By adding and subtracting terms, this can be rewritten as
\begin{align}
    \Psi(\hat{P}_n) - \Psi(P_0) 
    &= \bigl(\BP_n-P_0) D^*(\hat{P}_n) - \BP_n D^*(\hat{P}_n) + R(\hat{P}_n,P_0) \\
    &= \bigl(\BP_n-P_0) D^*(\hat{P}) - \BP_n D^*(\hat{P}_n) + \bigl(\BP_n-P_0) \Bigl\{D^*(\hat{P}_n) - D^*(\hat{P}) \Bigr\} + R(\hat{P}_n,P_0)
\end{align}
where $\hat{P}$ denotes the in-probability limit of $\hat{P}_n$, and we have used $PD^*(P) = 0$ for any $P \in \CM_0$.
Let $\bar{Q}_n$ and $g_n$ be the corresponding components of $\hat{P}_n$, let $\bar{Q}$ and $g$ be their in-probability limits (components of $\hat{P}$), and let $\bar{Q}_0$ and $g_0$ be the components of $P_0$.
Following~\cite{10.1093/biomet/asx053}, the above equation is written more precisely as
\begin{equation}\label{eq:Benkeser}
    \Psi(Q_n) - \Psi(Q_0) =
    \bigl(\BP_n-P_0) D^*(Q,g)
    - B_n(Q_n,g_n)
    + M_n(Q_n,Q,g_n,g)
    + R(Q_n,Q_0,g_n,g_0)
\end{equation}
where the two additional terms are the first-order bias term,
\begin{align}
    B_n(Q_n,g_n)
    = \BP_n D^*(Q_n,g_n),
\end{align}
and an empirical process term,
\begin{align}
    M_n(Q_n,Q,g_n,g)
    &= \bigl(\BP_n-P_0) \Bigl\{D^*(Q_n,g_n) - D^*(Q,g) \Bigr\}.
\end{align}
The empirical process term is $o_P(n^{-1/2})$ under mild conditions, see~\cite{10.1093/biomet/asx053,van1996weak}.
The rate of the second-order exact remainder term is determined by how fast the nuisance functions $\bar{Q}_0$ and $g_0$ are estimated.
Due to its double robust structure by Eq.~\ref{eq:point_treatment_is_double_robust}, this terms is also $o_P(n^{-1/2})$ provided the product of their rates is $n^{-1/2}$.
Sufficient rates of $n^{-1/4}$ for estimating both $\bar{Q}_0$ and $g_0$ can be achieved by some algorithms, such as the Highly Adaptive Lasso (HAL)~\citep{7796956}.
The first term, $\bigl(\BP_n-P_0) D^*(Q,g)$, is the average of $n$ independent and identically distributed copies of the random variable $D^*(Q,g)(O) - P_0 D^*(Q,g)$.
It is an important fact that this random variable has mean zero if either $Q = Q_0$ or $g = g_0$ (or both).
The first-order bias term $B_n(Q_n,g_n)$ can be substantial and should be dealt with in order to obtain an asymptotically linear estimator that attains the non-parametric efficiency bound given by the variance of the EIF $D^*(Q,g)$.
For the $k$-order interaction parameter of Eq.~\ref{eq:k-point}, combining Eq.~\ref{eq:gradient_point_treatment} with Lemma~\ref{lem:exact_remainder_k_point}, the first-order bias term is
\begin{align}\label{eq:bias_k_point}
    B_n(Q_n,g_n) = \frac{1}{n} \sum_{i=1}^n \sum_{s \in \{0,1\}^k} (-1)^{k - (s_1+\ldots+s_k)} 
    &\Biggl[\frac{\ID\{a_i=a(s)\}}{g_n(a_i,w_i)}\Bigl\{ y_i-\bar{Q}_n(a_i,w_i) \Bigr\} \\
    &+ \bar{Q}_n(a(s),w_i) - \frac{1}{n}\sum_{j=1}^n \bar{Q}_n(a(s),w_j) \Biggr] \\
    = \sum_{s \in \{0,1\}^k} (-1)^{k - (s_1+\ldots+s_k)} \Biggl[\frac{1}{n} \sum_{i=1}^n &\frac{\ID\{a_i=a(s)\}}{g_n(a_i,w_i)}\Bigl\{ y_i-\bar{Q}_n(a_i,w_i) \Bigr\} \Biggr],
\end{align}
where $o_i = (y_i,a_i,w_i) = (y_i,a_{j_1,i}, \ldots, a_{j_k,i},w_i)$ is the $i$th observed data point.

\subsection{Semi-parametric efficient estimators}

We discuss two general strategies to deal with the first-order bias term $B_n(Q_n,g_n)$ and obtain an asymptotically linear estimator of the $k$-order interaction parameter, provided the rate conditions for the empirical process term and second-order exact remainder are met.
The empirical process term is of the required rate $o_P(n^{-1/2})$ when $D^*(Q_n,g_n)$ belongs to a $P_0$-Donsker class with probability tending to one and $P_0\bigl\{D^*(Q_n,g_n)-D^*(Q,g)\bigr\}^2$ converges to zero in probability~\citep{10.1093/biomet/asx053}, or by employing sample splitting.
Throughout, we refer to these as the \emph{canonical} and the \emph{cross-validated} (CV) approaches. \\

The first approach, introduced by~\cite{PfanzaglWefelmeyer+1985+379+388}, is an infinite-dimensional generalisation of the Newton-Raphson method known as the \emph{one-step estimator} (OSE).
Here, the first-order bias term is simply added to the plugin estimator, yielding
\begin{equation}\label{eq:one_step_estimator}
    \hat{\psi}_n^{+} \coloneqq \hat{\psi}_n + B_n(Q_n,g_n).
\end{equation}
Provided the rate conditions on $M_n$ and $R_n$ hold, Eq.~\ref{eq:Benkeser} reads
\begin{equation}
    \hat{\psi}^+_n - \psi_0 = \BP_n D^*(Q,g) + o_P(n^{-1/2})
\end{equation}
demonstrating that the OSE is asymptotically linear with variance equal to the non-parametric efficiency bound given by the variance of the EIF $D^*(Q,g)$ provided %\footnote{Note that this condition implies $P_0 D^*(Q,g) = 0$.}
$Q=Q_0$ and $g=g_0$.
For a description of sampling behaviour when either $Q=Q_0$ or $g=g_0$, see Section~2.2 of~\cite{10.1093/biomet/asx053}.
For the average treatment effect, \ie, our target parameter with $k = 1$, the one-step estimator reduces to the Augmented Inverse Propensity Weighting (AIPW) estimator introduced by~\cite{AIPW94}.
While the one-step estimator is asymptotically efficient and straightforward to implement, its finite-sample performance can suffer since it is not a plug-in estimator and the first-order bias term may push the estimate outside of the target parameter's natural range.\\

The second approach, introduced by~\cite{vanderLaanRubin+2006}, instead updates the  fit $Q_n$ of $Q_0$ in an iterative procedure to a final fit $Q^{*}_n$ such that $B_n(Q^{*}_n,g_n) = 0$.
This updating step yields the Targeted Maximum-Likelihood, or Targeted Minimum Loss-based, Estimator (TMLE),
\begin{equation}\label{eq:TMLE}
    \hat{\psi}^{\textrm{tmle}}_n \coloneqq \Psi(Q^*_n).
\end{equation}
Similar to OSE, provided the rate conditions on $M_n$ and $R_n$ hold, Eq.~\ref{eq:Benkeser} reads
\begin{equation}\label{eq:TMLE_is_AL}
    \hat{\psi}^{\textrm{tmle}}_n - \psi_0 = \BP_n D^*(Q,g) + o_P(n^{-1/2})
\end{equation}
demonstrating that the TMLE is asymptotically linear with variance equal to the non-parametric efficiency bound given by the variance of the EIF $D^*(Q,g)$ when $Q=Q_0$ and $g=g_0$.
The sampling behaviour when either $Q=Q_0$ or $g=g_0$ is described in Section~2.2 of~\cite{10.1093/biomet/asx053}.
TMLE is a plugin estimator and hence enjoys finite-sample robustness properties~\citep{PorterGrubervanderLaanSekhon2011}.\\

The TMLE updating step consists in a linear regression (for continuous outcome $Y$) or logistic regression (for binary outcome $Y$) with offset, the initial fit $\bar{Q}_n(a,w)$, and covariate, the so-called clever covariate $H(g)(a,w)$ derived from the EIF $D^*(Q,g)$, namely
\begin{equation}\label{eq:fluctuations}
\begin{split}
    \bar{Q}_{n,\epsilon}(a,w) &= \bar{Q}_n(a,w) + \epsilon H(g)(a,w), \\
    \logit \bar{Q}_{n,\epsilon}(a,w) &= \logit \bar{Q}_n(a,w) + \epsilon H(g)(a,w)
\end{split}
\end{equation}
respectively.
The parameter $\epsilon$ is fitted by minimum loss-based estimation,
\begin{equation}\label{eq:epsilon_empirical_loss}
    \hat{\epsilon} = \arg\min_{\epsilon} \mathbb{P}_n \mathcal{L}\bigl\{\bar{Q}_{n,\epsilon}\bigr\},
\end{equation}
with respect to the squared loss function $\mathcal{L}\{f\}(O) = \bigl\{f(W,A)-Y\bigr\}^2$ for continuous outcome $Y$ or the log-loss function $\mathcal{L}\{f\}(O) = - Y \log f(W,A) - (1-Y)\log\bigl(1-f(A,W)\bigr)$ for binary outcome $Y$, where $f$ is a function of $(W,A)$.
The targeted update $\bar{Q}^*_n \equiv \bar{Q}_{n,\hat{\epsilon}}$ defines the TMLE of Eq.~\ref{eq:TMLE}, where we use the empirical distribution $\mathbb{Q}_W$ of $Q_W$ which is not updated as it is the nonparametric MLE.
These fluctuations, loss functions, and clever covariate are chosen so that the update $Q^*_n = (\bar{Q}^*_n, \mathbb{Q}_W)$ solves the EIF, \ie, it eliminates the first-order bias $B_n(Q^*_n,g_n) = 0$.
To see this, note
\begin{equation}
    \frac{d}{d\epsilon} \mathcal{L}\bigl\{\bar{Q}_{n,\epsilon}\bigr\}(O)\Big|_{\epsilon = 0} = H(g)(A,W)\bigl\{Y - \bar{Q}_n(A,W)\bigr\}.
\end{equation}
Since $\hat{\epsilon}$ is the minimiser of the empirical loss $\BP_n \CL\bigl\{\bar{Q}_{n,\epsilon}\bigr\}$ and $\bar{Q}^*_n \equiv \bar{Q}_{n,\hat{\epsilon}}$, it follows that
\begin{equation}
    0 
    = \frac{d}{d\epsilon} \mathbb{P}_n \mathcal{L}\bigl\{\bar{Q}^*_{n,\epsilon}\bigr\}\Big|_{\epsilon = 0}
    = \BP_n \Bigr\{ H(g_n)(A,W)\bigl\{Y - \bar{Q}^*_n(A,W)\bigr\} \Bigr\} \equiv B_n(Q^*_n,g_n)
\end{equation}
by Eq.~\ref{eq:bias_k_point}.
Here $\bar{Q}^*_{n,\epsilon}$ denotes the corresponding fluctuation of Eq.~\ref{eq:fluctuations} but with offset $\bar{Q}^*_n$ so that $\bar{Q}^*_{n,\epsilon = 0} \equiv \bar{Q}^*_n$.
Thus, for the $k$-point interaction $A \colon a(0) \to a(1)$ the clever covariate should be set to
\begin{equation}\label{eq:clever_covariate_int}
    H(g_n)(A,W) = \sum_{s \in \{0,1\}^k} (-1)^{k - (s_1+\ldots+s_k)} \frac{\ID\{A=a(s)\}}{g_n(A,W)},
\end{equation}
generalising the clever covariate for the average treatment effect, \ie, interaction with $k=1$.
This completely defines $\hat{\psi}_n^{\text{tmle}}$ of Eq.~\ref{eq:TMLE} and demonstrates it is asymptotically linear as per Eq.~\ref{eq:TMLE_is_AL}. \\

In finite samples, it has been observed in simulations~\citep{Sofrygin2017} that performance in terms of bias, variance, and coverage increases when fitting $\epsilon$ in Eq.~\ref{eq:fluctuations} via weighted regression, \ie, by minimising a weighted version of the empirical loss of Eq.~\ref{eq:epsilon_empirical_loss}.
Specifically, performance improves in the presence of near-positivity violations when the estimated propensity score $g_n(A,W)$ is close to zero.
This is particularly relevant in fields such as genetics, where DNA variants may be rare.
The weighted approach places the term $1/g_n(A,W)$ in the loss function,
\begin{equation}
    \mathcal{L}_{g_n}\bigl\{f\}(O) \equiv \frac{1}{g_n(A,W)}\mathcal{L}\bigl\{f\}(O),
\end{equation}
and removes it from the clever covariate.
The corresponding weighted fluctuations are
\begin{equation}\label{eq:fluctuations_weighted}
\begin{split}
    \bar{Q}^{(w)}_{n,\epsilon}(a,w) &= \bar{Q}_n(a,w) + \epsilon H'(a,w), \\
    \logit \bar{Q}^{(w)}_{n,\epsilon}(a,w) &= \logit \bar{Q}_n(a,w) + \epsilon H'(a,w)
\end{split}
\end{equation}
where the weighted clever covariate is
\begin{equation}\label{eq:clever_covariate_int_weighted}
    H'(A,W) = \sum_{s \in \{0,1\}^k} (-1)^{k - (s_1+\ldots+s_k)} \ID\{A=a(s)\}.
\end{equation}
The parameter $\hat{\epsilon}$ is now fitted by minimising the empirical weighted loss,
\begin{equation}\label{eq:epsilon_empirical_loss_weight}
    \hat{\epsilon} = \arg\min_{\epsilon} \mathbb{P}_n \mathcal{L}_{g_n}\bigl\{\bar{Q}_{n,\epsilon}\bigr\},
\end{equation}
yielding $\bar{Q}^*_n \equiv \bar{Q}^{(w)}_{n,\hat{\epsilon}}$ and $Q^*_n = (\bar{Q}^*_n,\mathbb{Q}_W)$.
The weighted TMLE (wTMLE) is finally defined as $\hat{\psi}_n^{\text{wtmle}} \equiv \Psi\bigl(Q^*_n\bigr)$.
For a proof that this choice of weighted fluctuations and weighted clever covariate also solves the EIF, \ie, satisfies $B_n(Q^*_n,g_n) = 0$, see Appendix~\ref{app:fluctuation}.
\\

Estimating the functions $(Q_n,g_n)$ with data-adaptive or ML algorithms whilst evaluating OSE and (w)TMLE on the same dataset may lead to decreased performance of the canonical estimators as it can affect the required rate conditions of the empirical process term $M_n$.
By using cross-validated (or sample-splitting) versions of these estimators, we can allow for complex algorithms to fit $(Q_n,g_n)$ on part of the data whilst maintaining performance of OSE and (w)TMLE by evaluating these estimators on a held-out part of the data, essentially treating the fits of $(Q_n,g_n)$ as fixed.\footnote{Specifically, using sample-splitting avoids the $P_0$-Donsker condition on algorithms used to fit $Q_n$ and $g_n$.}\\

More precisely, for $K$-fold sample splitting we split the data into $K \geq 2$ disjoint folds.
We write $k(i) \in \{1,2,\ldots,K\}$ for the fold to which sample $i$ belongs, and $-k(i)$ for the union of the $K-1$ remaining folds.
Similarly, we write $Q_n^{k}$ for an estimator of $Q_0 = (\bar{Q}_0,Q_W)$ using the samples in fold $k$ only, and we write $Q_n^{-k}$ for an estimator of $Q_0$ using the samples in all folds but $k$.
We use this notation to define the cross-validated versions of OSE and (w)TMLE.
The CV-OSE is
\begin{equation}\label{eq:one_step_estimator_CV}
    \hat{\psi}_n^{\cv, +} \coloneqq \sum_{k=1}^K \frac{n_k}{n} \Bigl\{ \Psi\left( \bar{Q}_n^{-k}, \mathbb{Q}_W^k \right) + \BP^k_n D^*\left(Q_n^{-k},g_n^{-k}\right) \Bigr\},
\end{equation}
where $n_k = |\{i \colon k(i) = k\}|$ denotes the number of samples in the $i$th fold, and $\BP_n^k$ denotes the empirical distribution of the $k$th fold.
Importantly, for the $K$ estimators in the sum, averages are taken over fold $k$ (via $\mathbb{Q}_W^k$ and $\BP_n^k$) whereas $Q_n^{-k}$ and $g_n^{-k}$ are estimated on all other folds.\\

For the CV-TMLE, the parameter $\hat{\epsilon}$ is defined in a pooled manner by the objective
\begin{equation}\label{eq:epsilon_empirical_loss_CV}
    \hat{\epsilon} 
    = \arg\min_{\epsilon} \sum_{k=1}^K \mathbb{P}^k_n \mathcal{L}\bigl\{\bar{Q}^{-k}_{n,\epsilon}\bigr\}
    = \arg\min_{\epsilon} \sum_{k=1}^K \sum_{i:k(i) = k} \mathcal{L}\bigl\{\bar{Q}^{-k}_{n,\epsilon}\bigr\}\bigl(O_i\bigr)
\end{equation}
where $\bar{Q}^{-k}_{n,\epsilon}$ is the respective path estimated on the folds in $-k$ and $\mathcal{L}$ is the respective loss function for continuous and binary outcomes.
Note that the loss of a sample $i$ in fold $k(i) = k$ is computed relative to nuisance functions fitted on all other folds $-k$.
The targeted update step now satisfies
\begin{equation}\label{eq:TMLE_update_CV}
    \bar{Q}^{*}_n(O_i) = \bar{Q}^{-k(i)}_{n,\hat{\epsilon}}(O_i)
\end{equation}
for all samples $i = 1, 2, \ldots, n$, and depends on the pooled estimate $\hat{\epsilon}$.
The CV-TMLE is defined as
\begin{equation}\label{eq:TMLE_CV}
    \hat{\psi}^{\cv,\textrm{tmle}}_n \coloneqq \sum_{k=1}^K \frac{n_k}{n} \Psi\left(\bar{Q}^*_n, \mathbb{Q}^k_W\right).
\end{equation}
While no longer a substitution estimator, the CV-TMLE respects the natural bounds of the target parameter because it is an average of substitution estimators.
The weighted CV-TMLE, which we denote by $\hat{\psi}_n^{\cv,\text{wtmle}}$, is defined analogously by using the weighted fluctuation and loss function.

% New page
%\clearpage

\subsection{Inference and hypothesis testing}
We obtain asymptotic Wald-type confidence intervals and hypothesis tests for interaction using the asymptotic normality of the interaction estimators constructed in the previous section.
If $\hat{\psi}_n$ is any of these estimators, with estimated nuisance functions $(Q_n,g_n)$, then
\begin{equation}
    \sqrt{n} \bigl(\hat{\psi}_n - \psi_0 \bigr) = \sqrt{n} \,\BP_n D^*(Q,g) + o_P(1) \rightsquigarrow \mathcal{N}\bigl(0, \Var D^*(Q,g)\bigr)
\end{equation}
by the Central Limit Theorem.
Here, recall that $(Q, g)$ are the in-probability limits of $(Q_n, g_n)$ respectively, and we assume either $Q = Q_0$ or $g = g_0$ (or both).
In practice, for the canonical estimators of Eqs.~\ref{eq:one_step_estimator} and~\ref{eq:TMLE}, we use the sample variance estimator $\hat{\sigma}_n^2$ built from $(Q_n,g_n)$ to approximate the variance of the EIF, namely
\begin{equation}\label{eq:variance_iid}
    \hat{\sigma}_n^2 = \BP_n \bigl\{D^*(Q_n,g_n)\bigr\}^2 = \frac{1}{n} \sum_{i=1}^n \Bigl\{D^*(Q_n,g_n)(o_i)\Bigr\}^2.
\end{equation}
For the cross-validated versions of the OSE and (weighted) TMLE of Eq.~\ref{eq:one_step_estimator_CV} and~\ref{eq:TMLE_CV}, we approximate the variance of the EIF with cross-validated sample variance estimators $\hat{\sigma}_{\text{cv},n}^2$ built from the nuisance functions $(Q^{-k}_n,g^{-k}_n)$ estimated on all folds but $k$ for $k = 1, \ldots, K$, namely
\begin{equation}\label{eq:cv_variance_iid}
   \hat{\sigma}_{\text{cv},n}^2 = \sum_{k=1}^K \frac{n_k}{n} \BP^k_n \bigl\{ D^*\left(Q_n^{-k},g_n^{-k}\right) \bigr\}^2 = \frac{1}{n} \sum_{k=1}^K \sum_{i : k(i) = k} \Bigl\{ D^*\left(Q_n^{-k},g_n^{-k}\right)(o_i) \Bigr\}^2.
\end{equation} 
Both $\hat{\sigma}_n^2$ and $\hat{\sigma}_{\text{cv},n}^2$ are consistent estimators of $\Var D^*(Q_0,g_0)$ if the rate conditions on $M_n$ and $R_n$ hold.
We obtain asymptotically valid $(1-\alpha) \times 100\%$ Wald-type confidence intervals of the form
\begin{equation}\label{eq:CI_wald}
    \widehat{\ConfI}_n = \left( \hat{\psi}_n - z_{1-\alpha/2}\frac{\hat{\sigma}_n}{\sqrt{n}}, \hat{\psi}_n + z_{1-\alpha/2}\frac{\hat{\sigma}_n}{\sqrt{n}} \right),
\end{equation}
where $z_{\beta}$ denotes the $\beta$-quantile of the standard normal distribution.
Similarly, under the null hypothesis of no interaction, $H_0 \colon \psi_0 = 0$, we can use $\sqrt{n} \hat{\psi}_n / \hat{\sigma}_n \sim \CN(0,1)$ to test for interaction.\\

In a typical GWAS, the genetic effect size of a DNA variant on phenotype is estimated as a coefficient in a Linear Mixed Model (LMM).
To bound false positives due to multiple testing across many variants, either a Bonferroni correction, bounding the Family-Wise Error Rate (FWER), or a Benjamini--Hochberg correction, bounding the False Discovery Rate (FDR), is employed.
Due to the linearity assumption of the LMM, this means that each DNA variant contributes a single test to the multiple testing procedure.
In contrast, our first-order genetic effect size of Eq.~\ref{eq:k-point} for $k = 1$, estimates both allelic effects separately, yielding twice as many tests.
For a $k$-order interaction, this constitutes $2^k$ times as many tests\footnote{In the literature, a baseline cut-off is typically placed on the positivity of the propensity score or, in genetics, on the minor allele frequency of a DNA variant (independent of covariate strata).
Instead, we are conservative and simply do not test genotype changes for which the positivity threshold is not met. This means that there are at most $2^k$ times as many tests in our procedure, but in practice often fewer.}, increasing the testing burden.\\

It is often of interest to understand whether a DNA variant at a locus, or a group of $k$ DNA variants in an interaction, are significantly associated with a trait regardless of the specific genotype at the locus or loci.
However, joint testing is required because shared use of the heterozygote genotype at a locus makes the statistical tests dependent.
To construct asymptotically valid confidence intervals and hypothesis tests, we appeal to multi-variate normality of our semi-parametric estimators and use Hotelling's T-squared statistic $(T^2$) as multi-estimand generalisation of the one-dimensional t-statistic.
More precisely, let $\psi = (\psi_1,\ldots,\psi_p)$ be a $p$-dimensional estimand with canonical gradient
\begin{equation}
    D^*_{\psi} = (D^*_{\psi_1},\ldots,D^*_{\psi_p}).
\end{equation}
Let $\hat{\psi}_n = (\hat{\psi}_{1,n}, \ldots, \hat{\psi}_{p,n})$ be a vector of asymptotically linear estimators with EIF $D^*_{\psi}$, then
\begin{equation}
    \sqrt{n} \bigl(\hat{\psi}_n - \psi_0 \bigr) = \sqrt{n} \, \BP_n D^*(Q,g) + o_P(1) \rightsquigarrow \mathcal{N}\bigl(0, \Var D^*_{\psi}(Q,g)\bigr)
\end{equation}
by the multi-variate Central Limit Theorem.
Here $\Sigma \equiv \Var D^*_{\psi}(Q,g)$ denotes the $p \times p$-dimensional covariance matrix of $D^*_{\psi}(Q,g)$ with in the $jk$ entry the element
\begin{equation}
    \Sigma_{jk} = \Cov\bigl(D^*_{\psi_j},D^*_{\psi_k}\bigr).
\end{equation}
In practice, we use the sample covariance matrix $\hat{\Sigma}_n$ built from the estimated functions $(Q_n,g_n)$:
\begin{equation}
    \bigl(\hat{\Sigma}_n\bigr)_{jk} 
    = \BP_n \bigl\{D^*_{\psi_j}(Q_n,g_n)D^*_{\psi_k}(Q_n,g_n)\bigr\}
    = \frac{1}{n} \sum_{i=1}^n \Bigl\{ D^*_{\psi_j}(Q_n,g_n)(o_i) \cdot D^*_{\psi_k}(Q_n,g_n)(o_i) \Bigr\}.
\end{equation}
This is a consistent estimator of $\Sigma$ provided the rate conditions on $M_n$ and $R_n$ hold for all $\hat{\psi}_{i,n}$.
To jointly test if any of the parameters are significant, we use Hotelling's T-squared statistic defined as
\begin{equation}
    t^2_n = n\bigl(\hat{\psi}_n - \psi_0\bigr) \hat{\Sigma}^{-1}_n \bigl(\hat{\psi}_n - \psi_0\bigr)^T.
\end{equation}
It follows a scaled $F$-distribution, namely $t^2_n \sim \frac{p(n-1)}{n-p} F_{p,n-p}$.
This can be used to construct Wald-type confidence regions, obtain asymptotically valid $p$-values, and perform hypothesis tests against the joint null $H_0 \colon \psi_0 = 0$ that all $p$ one-dimensional parameters are zero simultaneously.\\

As an example, suppose $k=1$ and the treatment variable $A \equiv A_{j_1}$ has three treatment levels denoted by $\{0, 1, 2\}$.
In genomics, $\{0,1,2\}$ could represent the genotypes $\{TT, TC, CC\}$ at a biallelic locus of the genome, where the treatment level corresponds to the number of copies of the minor allele $C$.
Using the definition of $k$-order interaction in Eq.~\ref{eq:k-point}, there are three distinct target parameters, namely $\psi_1 \equiv \Psi_{0 \to 1}$, $\psi_2 \equiv \Psi_{1 \to 2}$, $\Psi_{0 \to 2}$ corresponding to the changes $a \colon 0 \to 1$, $a \colon 1 \to 2$, and $a \colon 0 \to 2$ respectively.
Consider the two-dimensional parameter $\psi = (\psi_1, \psi_2)$.
We seek to understand whether the addition of a single minor allele $(0 \to 1$ or $1 \to 2$) has a significant effect on outcome $Y$, correcting for covariates $W$.
The Hotelling statistic is $t^2_n = n (\hat{\psi}_{1,n}, \hat{\psi}_{2,n}) \hat{\Sigma}^{-1}_n (\hat{\psi}_{1,n}, \hat{\psi}_{2,n})^T$.
Its corresponding confidence regions are ellipses with dimensions and orientation controlled by the constituent signal-to-noise ratios and the correlation of $\hat{\psi}_{1,n}$ and $\hat{\psi}_{2,n}$.\\

More generally, suppose we have an asymptotically linear estimator $\hat{\psi}_n$ of our p-dimensional estimand $\psi = (\psi_1,\ldots,\psi_p)$, and we are interested in an estimate of $f(\psi)$ for some differentiable function $f \colon \BR^p \to \BR^q$.
Since $\sqrt{n} \bigl(\hat{\psi}_n - \psi_0\bigr) \rightsquigarrow \mathcal{N}(0,\Sigma)$, the functional delta method~\citep{van2000asymptotic} states that $f(\hat{\psi}_n)$ is also an asymptotically linear estimator of $f(\psi_0)$, and that
\begin{equation}\label{eq:delta_method}
    \sqrt{n} \Bigl(f(\hat{\psi}_n) - f(\psi_0) \Bigr)
    \rightsquigarrow \mathcal{N}\Bigl(0, \bigl(\nabla f(\psi_0)\bigr) \Sigma \bigl(\nabla f(\psi_0)\bigr)^T\Bigr). 
\end{equation}
Here $\nabla f(\psi_0) \colon \BR^q \to \BR^p$ denotes the (transposed) Jacobian matrix of first-order partial derivatives.
An important example in genetics is the non-linear allelic effect estimand $\psi_{\Delta} = \psi_2-\psi_1$ contrasting the population effect of adding a second minor allele ($1 \to 2$) versus adding the first minor allele ($0 \to 1$).
This estimator is obtained by applying the function $f \colon \BR^2 \to \BR$, $f(x_1,x_2) = x_2-x_1$ to the estimand $\psi = (\psi_1,\psi_2)$ in the example above.
The delta method of Eq.~\ref{eq:delta_method} yields
\begin{equation}\label{eq:nonlinear_allelic_difference}
    \sqrt{n} \bigl( \hat{\psi}_{\Delta,n} - \psi_{\Delta,0} \bigr)
    \rightsquigarrow \mathcal{N}\Bigl(0, \Var D^*_{\psi_1} - 2\Cov(D^*_{\psi_1},D^*_{\psi_2}) + \Var D^*_{\psi_2}\Bigr), 
\end{equation}
incorporating the dependence of $\hat{\psi}_{1,n}$ and $\hat{\psi}_{2,n}$ via the covariance of their influence functions.
The non-linear allelic effect estimator $\hat{\psi}_{\Delta}$ is directly available in the \texttt{TMLE.jl} package and \texttt{TarGene} software.
Automatic differentiation functionalities in \texttt{Julia} allow the user to specify their desired $p$-dimensional estimand and differentiable function $f$.

% New page
%\clearpage

\subsection{Sieve plateau variance estimators}
\label{sec:svp}
%Take PLOS section and condense.
%State formal Theorem (+ ref) Davies-VdL

Biobank cohorts consist of participants who are, to some extent, related due to ancestry or kinship.
TarGene accounts for this population dependence structure by appropriately adjusting variance estimates of effect sizes and interactions via Sieve Plateau (SP) variance estimators~\citep{SievePlateau}.
These estimators are based on a notion of genetic similarity between participants as encoded in the Genetic Relationship Matrix (GRM).
Population-dependence induces dependence among the variables $O_i$, and thus the influence functions of the target parameter of interest via genetic similarity.
Such genetic similarity can occur on a subpopulation level due to ancestry (\eg, being white Irish), or on an individual level due to kinship (\eg, parents, children, cousins).
Moreover, genetically similar individuals may share diet and environment inducing further dependence~\citep{GxE_GeographicCorrelations}.\\

The genetic similarity of two individuals $i$ and $j$ is quantified by the sample correlation coefficient $G_{ij}$ between their (centred and scaled) DNA variants.
Together, these coefficients form the GRM, denoted $G$, of size $N \times N$ where $N$ is the number of individuals in the population~\citep{henderson1975use}. 
More precisely, given a set of $R$ variants, we have
\begin{equation}\label{eq:GRM}
    G_{ij} = \frac{1}{R-1} \sum_{k = 1}^{R} \frac{(s_{ik} - 2p_k)(s_{jk} - 2p_k)}{2p_k(1-p_k)}.
\end{equation}
Here $s_{ik} \in \{0,1,2\}$ denotes the number of copies of the reference allele for individual $i$ at variant $k$, and $p_k \in (0,1)$ denotes the frequency of the reference allele at variant $k$ over the population of $N$ individuals.
In particular, the population average of $s_{ik}$ equals twice the reference allele frequency at variant k, \ie, $2p_k$ (one for each strand copy), so
\begin{equation}
    \frac{1}{N}\sum_{i = 1} ^{N} s_{ik} = 2p_k.
\end{equation}
Therefore, $\tilde{s}_{ik} = s_{ik} - 2p_{k}$ is the zero-centred count of the number of copies of the reference allele of individual $i$ at variant k.
Considered as a random variable, $\tilde{s}_{ik}$ can take on three values. 
Assuming reference alleles are sampled binomially with mean frequency $p_k$ (\ie, approximately satisfy Hardy--Weinberg equilibrium), the standard deviation of $\tilde{s}_{ik}$ equals $\sqrt{2p_k(1-p_k)}$.
This explains the additional factor in Eq.~\ref{eq:GRM} that scales the variables $\tilde{s}_{ik}$ and $\tilde{s}_{jk}$ so as to have unit variance.
Finally, note that the GRM depends on the set of $R$ selected variants.
These variants are chosen among genotyped (not imputed) variants which, in addition, are pruned not to be in LD with one another. \\

In TarGene, we neither assume individuals are independent nor impose the strong modelling assumptions placed on dependence by an LMM.
Instead, we incorporate the genetic dependence of individuals model-independently using an approach drawn from~\citep{SievePlateau}.
This approach generalises Eq.~\ref{eq:variance_iid} for the variance of the $k$-order interaction parameter of Eq.~\ref{eq:k-point} by constructing SP variance estimators that incorporate genetic dependence of individuals.\\
%These estimators result in valid confidence intervals and, ultimately, realistic and valid p-values having correctly accounted for population stratification.

We now illustrate how data dependence impacts the variance estimate of Eq.~\ref{eq:variance_iid}.
Since individuals $i$ and $j$ are generally dependent, their data $O_i = (W_i,A_{1,i}, \ldots, A_{m,i}, Y_i)$ and $O_j$, as well as their corresponding influence curves $D^{*}_{P}(O_i)$ and $D^{*}_{P}(O_j)$, are also in general dependent.
The problem arises since for two random variables $X_1$ and $X_2$ the variance of their sum is \emph{not} in general equal to the sum of their variances.
The difference is exactly twice the covariance of $X_1$ and $X_2$,
\begin{equation}
    \Var(X_1\pm X_2) = \Var(X_1) + \Var(X_2) \pm 2\Cov(X_1,X_2).
\end{equation}
The impact of this difference may be large depending on the size of the covariance $\Cov(X_1,X_2)$; note that the covariance of two random variables vanishes when they are independent.
Thus, rather than Eq.~\ref{eq:variance_iid}, the true variance of the OSE or (w)TMLE semi-parametric estimators $\hat{\psi}_n$ is given by
\begin{equation}\label{eq:TL_effect_size_variance_dependent}
    \sigma_n^2 %= \Var\Bigl[\sqrt{n} \bigl( \Psi(\hat{Q}^{*}_n) - \Psi(P_0) \bigr)\Bigr] 
    = \Var\bigl[ \BP_n D^{*}(Q,g) \bigr]
    = \frac{1}{n} \sum_{i=1}^{n} \sum_{j=1}^n \Cov\Bigl( D^{*}(Q,g)(O_i), D^{*}(Q,g)(O_j) \Bigr).
\end{equation}
As before, in practice we use the sample variance estimator $\hat{\sigma}_n^2$ built from the estimated functions $(Q_n,g_n)$ to approximate the covariance of the EIF.
The distinction between Eq.~\ref{eq:variance_iid} and Eq.~\ref{eq:TL_effect_size_variance_dependent} is only relevant for sufficiently genetically similar individuals.
SP variance estimators use a cut-off $\tau$ for the genetic similarity between individuals, and set the covariance to zero if individuals are sufficiently genetically dissimilar.
This produces a family of variance estimates, $\hat{\sigma}_{n}^2(\tau)$.
Under appropriately weak dependence (see~\cite[Theorem~1]{SievePlateau}), the true variance is obtained where the function $\tau \mapsto \hat{\sigma}_n^2(\tau)$ plateaus.\\

To construct SP variance estimators, we proceed as follows:
\begin{enumerate}
    \item Using the GRM, we define genetic dissimilarity between individuals $i$ and $j$ as
    \begin{equation}\label{eq:genetic_distance}
        d(i,j) = 1 - G_{ij},
    \end{equation}
    where $G_{ij}$ is the sample correlation coefficient of Eq.~\ref{eq:GRM} quantifying the genetic dependence between individuals $i$ and $j$.
    Since correlation is bounded, $|G_{ij}| \leq 1$, the genetic dissimilarity is non-negative and never larger than two, \ie, $0 \leq d(i,j) \leq 2$.
    Biologically, if two individuals $i$ and $j$ have identical DNA variants, they are fully correlated, $G_{ij} = +1$, and thus have zero genetic distance, $d(i,j) = 0$, as expected.
    \item Given a value for the cut-off $\tau \in [0,1]$, we define a SP variance estimator as
    \begin{equation}\label{eq:SP_variance_estimator}
        \hat{\sigma}^2_n(\tau) = \frac{1}{n} \sum_{i=1}^{n} \sum_{j=1}^{n} \ID\{ d(i,j) \leq \tau\} \cdot D^{*}(Q_n,g_n)(o_i) D^{*}(Q_n,g_n)(o_j).
    \end{equation}
    Here, the term $\ID\{ d(i,j) \leq \tau\}$ equals $1$ if the genetic dissimilarity between individuals $i$ and $j$ is at most $\tau$, \ie, $d(i,j) \leq \tau$, and it equals $0$ otherwise.
\end{enumerate}
The biological interpretation of these estimators is as follows.
The correlation between the influence functions of individuals $i$ and $j$, estimated by the term $D^{*}(Q_n,g_n)(o_i) D^{*}(Q_n,g_n)(o_j)$, is taken into account only if the genetic dissimilarity between individuals $i$ and $j$ is at most $\tau$.
Thus, the SP variance estimator $\hat{\sigma}^2_n(0)$ (for $\tau = 0$) assumes all individuals are \emph{independent}.
By increasing $\tau$, we first take the covariance between strongly genetically dependent individuals into account for low $\tau$, and then incorporate the covariance of more weakly dependent individuals as $\tau$ increases to $\tau = 1$.
\begin{enumerate}
    \item[(3)] We construct the variance estimator $\hat{\sigma}^2_n(\tau)$ for a number of values of the cut-off $\tau$, \eg, $\tau = 0,\frac{1}{k}, \frac{2}{k}, \ldots, \frac{k-1}{k}, 1$.
    Then we fit the function $\tau \mapsto \hat{\sigma}^2_n(\tau)$ and select the value of $\tau$ where the function plateaus, call it $\tau_0$.
    \item[(4)] The selected Sieve Plateau variance estimate is $\hat{\sigma}^2_n(\tau_0)$.
\end{enumerate}
Under mild assumptions \cite[Theorem~1]{SievePlateau}, the distribution of the effect size estimate $\hat{\psi}_n$ of variant $A$ on outcome $Y$ (or the effect of a $k$-order interaction on outcome) is asymptotically normal, and the SP variance estimator allows for the construction of an approximate $(1-\alpha) \times 100\%$ Wald-type confidence interval for the estimate $\hat{\psi}_n$ in the usual way, namely
\begin{equation}\label{eq:TL_confidence_interval_SP}
    \widehat{\ConfI}_n = \left( \hat{\psi}_n - z_{1-\alpha} \sqrt{\frac{\hat{\sigma}^2_n(\tau_0)}{n}}, \hat{\psi}_n + z_{1-\alpha} \sqrt{\frac{\hat{\sigma}^2_n(\tau_0)}{n}} \right).
\end{equation}
From here, we obtain realistic p-values correctly accounting for population dependence.\\

Similarly, we need to take into account population dependence in order to obtain a realistic estimate of the variance on $\hat{\psi}_{\Delta,n} = \hat{\psi}_{2,n} - \hat{\psi}_{1,n}$ in Eq.~\ref{eq:nonlinear_allelic_difference}.
If this difference of effect sizes of a DNA variant $A$ on outcome $Y$ is significant, then the effect of an additional allelic copy is non-linear.
We construct an SP estimator for the variance on this difference by following steps (1)--(4) above, with the exception of appropriately generalising Eq.~\ref{eq:SP_variance_estimator} in step (2) as follows.
Given a value for the cut-off $\tau \in [0,1]$, the estimator is
\begin{equation}\label{eq:SP_variance_estimator_difference}
\begin{split}
    \hat{\delta}^2_n(\tau) = \frac{1}{n} \sum_{i=1}^{n} \sum_{j=1}^{n} \ID\{ d(i,j) \leq \tau\} \Bigl[ &\hat{D}^{*}_{\psi_1}(o_i) \hat{D}^{*}_{\psi_1}(o_j) - \hat{D}^{*}_{\psi_1}(o_i) \hat{D}^{*}_{\psi_2}(o_j) \\ 
    -& \hat{D}^{*}_{\psi_2}(o_i) \hat{D}^{*}_{\psi_1}(o_j) + \hat{D}^{*}_{\psi_2}(o_i) \hat{D}^{*}_{\psi_2}(o_j) \Bigr].
\end{split}
\end{equation}
Here we have used the short-hand $\hat{D}^{*}_{\psi_i} = D^{*}_{\psi_i}(Q_n,g_n)$ for the influence function of estimand $\psi_i$ evaluated at the estimated nuisance functions $(Q_n,g_n)$ for the OSE $\hat{\psi}_n^{+}$ or the final estimated nuisance functions (here denoted similarly by $(Q_n,g_n)$) for the (weighted) TMLE $\hat{\psi}_n^{\text{(w)tmle}}$.

%%%%%%%%%%%%%%%%%%%%%%%%%%%
%%%%%%% Simulations %%%%%%%
%%%%%%%%%%%%%%%%%%%%%%%%%%%
\input{simulations}

%%%%%%%%%%%%%%%%%%%%%%%%%%%
%%%%%% Applications %%%%%%%
%%%%%%%%%%%%%%%%%%%%%%%%%%%
\input{applications}

%%%%%%%%%%%%%%%%%%%%%%%%
%%%%%%% Software %%%%%%%
%%%%%%%%%%%%%%%%%%%%%%%%
\input{software}

%%%%%%%%%%%%%%%%%%%%%%%
%%%%% Discussion %%%%%%
%%%%%%%%%%%%%%%%%%%%%%%

\newpage

\section{Discussion}\label{sec:Discussion}

In this paper, we have introduced Targeted Genomic Estimation (TarGene), a method based on targeted semi-parametric estimation theory \citep{OS_Pfanzagl,MR2867111} following the TL roadmap for the estimation of genetic effects of single variants and interactions.
TarGene offers a number of distinct advantages over commonly employed LMM approaches in population genetics as it avoids model-misspecification bias, produces asymptotically normal and efficient estimates, and is doubly-robust.
We remark that linear models are a standard part of the SL library of TarGene so if statistical inference using a linear model is correct, then TarGene will choose the model in a data-driven manner.
Furthermore, due to the flexibility of its SL libraries and the TMLE step, it can be readily applied to more ancestry heterogeneous biobanks such as \emph{All of US}~\citep{AllOfUs} or the \emph{Million Veterans Program}~\citep{GAZIANO2016214}, as well as more strongly inter-related cohorts such as island communities.
We have demonstrated in extensive realistic and data-adaptive simulations that TarGene achieves nominal coverage and control of type I error, provided the minor allele frequency is bounded from below by $0.01$.\\

Whilst TarGene's run time is slower than some other GWAS estimation methods, the strength of TarGene lies in bespoke analyses of effect sizes and interactions among targeted variants of interest, providing mathematically guaranteed coverage of the ground truth.
The run time of this workflow depends on how precise and unbiased a researcher wishes to be regarding the answer to their question of interest.
For researchers interested in applying TarGene for genome-wide studies across multiple traits, we note that there is a trade-off between computational speed and guaranteeing ground-truth coverage of estimates. 
In such a scenario we therefore recommend equipping the SL with computationally light algorithms only, such as a linear model, GLMnet or LMM, reducing the cross-validation burden significantly, before running the TMLE step. The lightest version of TarGene can be run in 10 hours on a HPC cluster, for a GWAS with~600,000 variants (Table~\ref{tab:gwas}).
Significant hits from an initial light run can then be verified by a more comprehensive TarGene run with a stacked SL library. 
Taking advantage of the approach described by~\cite{TuglusvanderLaan+2009} nevertheless allows the researcher to control the final FDR of this two-stage procedure. We also recommend applying the SP variance correction to estimates with p-values near the FDR threshold to ensure their robustness. 
We also remark that, although the SP variance estimator requires reading the GRM into memory, in comparison to LMMs, TarGene does not require the memory-intensive inversion of the GRM.
This inversion step may be prohibitive for UKB-scale cohort sizes, depending on institutional resources.\\

Finally, TarGene estimators address any statistical gap due to model misspecification as well as the causal gap due to population stratification, they do not currently close the causal gap due to LD.
While attempts have been made in the statistics literature to address the causal gap due to LD through fine-mapping, \eg, SuSiE~\citep{10.1111/rssb.12388,10.1371/journal.pgen.1010299}, and KnockOffGWAS~\citep{Sesiae2105841118}, these methods do not close the statistical gap due to parametric assumptions or are unable to report (interaction) effect estimates, respectively.
These quantities are essential for explaining how variants, via biological mechanisms and regulatory functions, modify a trait or disease risk.
A unified method that closes both causal and statistical gaps in genomic medicine has yet to be developed.\\

In future work, we will investigate various collaborative TML estimators (CTMLE)~\citep{vanderLaanGruber+2010,doi:10.1177/0962280217729845} to reduce the causal gap due to LD. 
CTMLE has been applied successfully in situations with many potential confounders, see \eg,~\cite{Mireille2016, Pirracchio2018}.
Additionally, we plan to explore non-linearities~\cite{Neale-nonlinearity-science} in variant allelic copies on trait using TarGene.
We will also investigate the contribution of epistatic interactions of specific variants on various polygenic traits for a variety of biological mechanisms.

\section{Acknowledgements} 

This research has been conducted using the UK Biobank Resource under Application Number 53116.
OL was supported by the United Kingdom Research and Innovation (grant EP/S02431X/1), UKRI Centre for Doctoral Training in Biomedical AI at the University of Edinburgh, School of Informatics.
BRH was supported by the Health Data Research UK \& The Alan Turing Institute Wellcome PhD Programme in Health Data Science (Grant Ref: 218529/Z/19/Z).
MvdL is supported by NIH grant \verb|R01AI074345|.
CPP was funded by the MRC (\verb|MC_UU_00007/15|).
AK was supported by the XDF Programme from the University of Edinburgh and Medical Research Council (\verb|MC_UU_00009/2|), and by a Langmuir Talent Development Fellowship from the Institute of Genetics and Cancer, and a philanthropic donation from Hugh and Josseline Langmuir.
The authors thank John Ireland for guidance on running TarGene on the University of Edinburgh high-performance computing cluster (Eddie).
The authors gratefully acknowledge All of Us participants for their contributions, without whom this research would not have been possible. We also thank the National Institutes of Health’s All of Us Research Program for making available the participant data examined in this study.
For the purpose of open access, the author has applied a creative commons attribution (CC BY) licence to any author accepted manuscript version arising.

\section{Data availability}
This study used data from the All of Us Research Program’s Controlled Tier Dataset v7, available to authorized users on the Researcher Workbench.

\section{Competing interests}
No competing interests declared.

\clearpage

%%%%%%%%%%%%%%%%%%%
%%%%% Appendix %%%%
%%%%%%%%%%%%%%%%%%%
\include{appendix}

%%%%%%%%%%%%%%%%%%%%%%%%
%%%%% Supplementary %%%%
%%%%%%%%%%%%%%%%%%%%%%%%
\include{supplementary}

\newpage

%\bibliographystyle{plainnat}
%\bibliographystyle{plain}
%\bibliographystyle{unsrt} % references in order of appearence

% references
\newpage
\bibliographystyle{apalike} 
\bibliography{refs.bib}

\end{document}

%% file: macros.tex
\usepackage{tikz-cd}
\usepackage{amsthm}

\usepackage{dsfont}
\newcommand{\ID}{\mathds{1}}

\newtheorem{prop}{Proposition}[section]

\newtheorem{lemma}[prop]{Lemma}

\theoremstyle{definition}

\usepackage{caption}
\usepackage{subcaption}
\usepackage{verbatim}

\usepackage[font={small}]{caption}

\usepackage[foot]{amsaddr}

\def\ie{{\it i.e.}}
\def\eg{{\it e.g.}}

\DeclareMathOperator{\cv}{cv}

\DeclareMathOperator{\Var}{Var}
\DeclareMathOperator{\Cov}{Cov}
\DeclareMathOperator{\Bias}{Bias}
\DeclareMathOperator{\MSE}{MSE}

\DeclareMathOperator{\ConfI}{CI}

\newcommand{\BE}{{\mathbb{E}}}

\newcommand{\BP}{{\mathbb{P}}}

\newcommand{\BR}{{\mathbb{R}}}

\newcommand{\CL}{{\mathcal L}}
\newcommand{\CM}{{\mathcal M}}
\newcommand{\CN}{{\mathcal N}}

\newcommand{\CT}{{\mathcal T}}

\DeclareMathOperator{\logit}{logit}
\DeclareMathOperator{\expit}{expit}

%% file: simulations.tex
\section{Simulations}\label{sec:Simulations}
%(1) Model-misspecification simulation import.
%(2) Interaction simulation (2-point and 3-point). Ask Olivier to run this + place on the TMLE.jl page.

While semi-parametric estimators are theoretically asymptotically optimal, asymptotic regimes are not necessarily achieved in practice, even for large sample sizes because events may be rare.
This is particularly prevalent in population genetics, where some genetic variants and traits are found in less than $1\%$ of individuals.
In the absence of finite sample guarantees, simulation studies provide an effective way to validate statistical methods.\\

In population genetics, these simulations are often based on simple parametric models for which ground truth values can be obtained via algebraic formulas.
However, it has been recognised that these simulations lack important features of the true generating process~\citep{schuler2017synth, parikh2022validating, li2022evaluating}.
As a result, conclusions obtained through simple parametric simulations may not generalise well to real-world data.
More realistic simulations can be constructed by modelling the generating process with flexible data-adaptive generative models.
While this approach does not yield a closed form expression of the ground truth, an arbitrarily precise estimate can be obtained via Monte Carlo sampling.\\

Here we use data-adaptive generative models to analyse the performance of semi-parametric estimators through two types of simulations.
The first simulation examines conditions under which these estimators appropriately control type I error rates when the null hypothesis of no effect is true. This is referred to as the \emph{null simulation}.
In the second simulation, we model the data-generating process using flexible neural network-based models.
We then verify that the conditions identified in the first simulation provide nominal coverage and estimate the power of discovery. This is referred to as the \emph{realistic simulation}.
Both simulations utilise the entire UK Biobank dataset, rather than a homogeneous subsample (\eg, white British), to retain as many structural dependencies as possible.

\subsection{Null Simulation}

In the null simulation, we analyse the behaviour of the estimators when the null hypothesis of no variant effect on trait is true, \ie, the true effect size is zero.
The rationale behind this simulation is that most genetic variants are believed to have no or little effect~\citep{rands20148,Omnigenic_model}, and it is thus of particular importance that the type I error rate be controlled appropriately.
Since we are interested in the effects of genetic variants on traits, the data generating process must satisfy $Y \indep V_j$ for all $j = 1, \ldots, p$ variants, where $Y$ is a given trait and $V_j$ denote the treatment variables called $A_j$ in the previous sections.
In practice, we enforce a stronger condition, where all variables are pairwise independent.
This is done by independently drawing $n$ samples with replacement from the empirical marginal distribution of each variable.
The only exception is that PCs and additional covariates $C$ are sampled jointly.
This generating process, presented in Fig.~\ref{fig:samplers} (left), hence preserves many characteristics of the original dataset, while resulting in the true null hypothesis of no genetic effect on outcome.

\begin{figure}[H]
    \centering
    \includegraphics[scale=0.6]{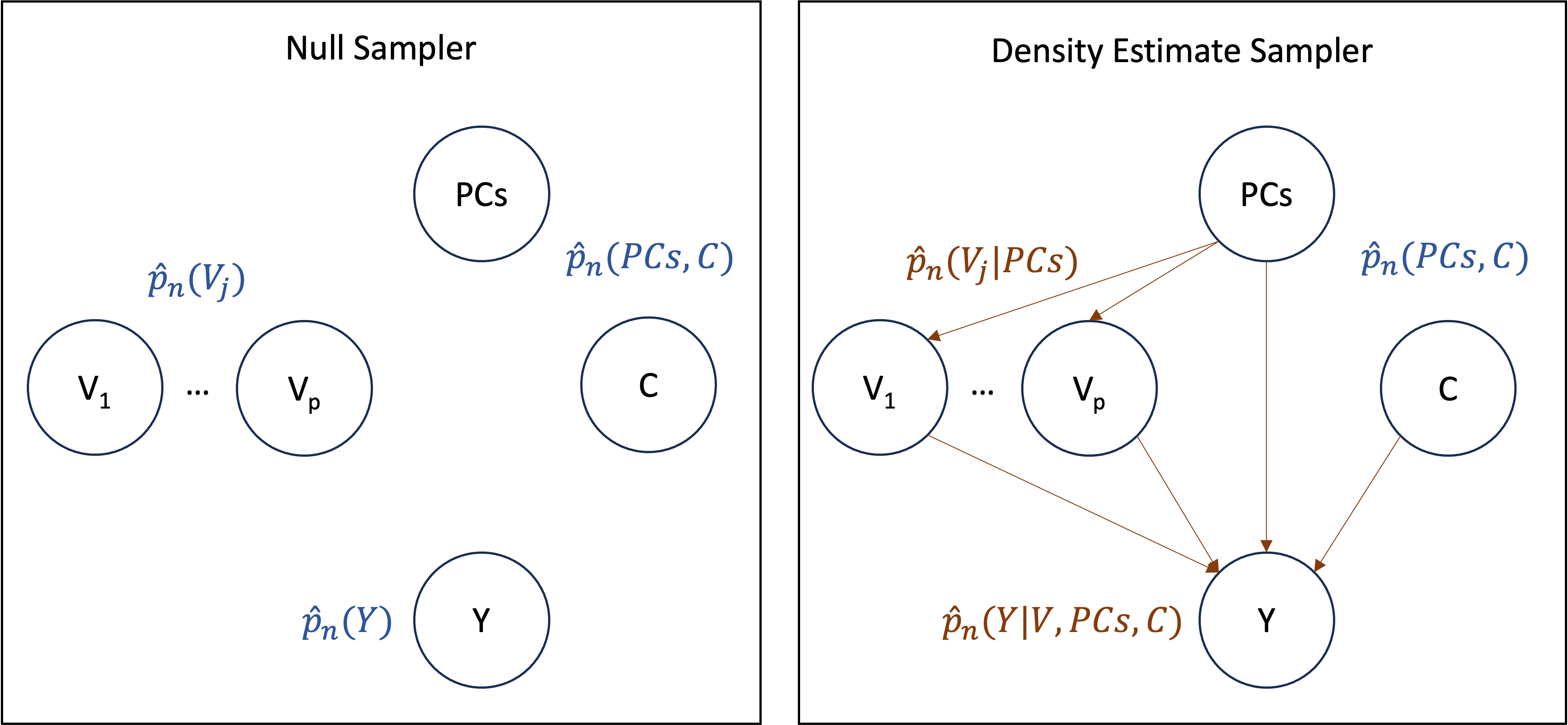}
    \caption{\textbf{Generating processes of simulation studies.} Empirical marginal distributions are coloured in blue while learnt conditional densities are coloured in orange. In both cases $(PCs, C)$ are sampled jointly using the empirical marginal distribution.
    \textbf{Left:} The null sampler independently samples from the empirical marginal distributions of each $Y$, $V_j$. This results in the theoretical null hypothesis of no effect.
    \textbf{Right:} The density estimate sampler proceeds via ancestral sampling, first each $V_j$ is sampled from $\hat{P}_n(V_j|PCs)$, then $Y$ is sampled from $\hat{P}_n(Y|\mathbf{V}, PCs, C)$. The various causal effects can then be approximated via Monte Carlo sampling using $\hat{P}_n\bigl(Y|do(\mathbf{V}), PCs, C\bigr)$.}
    \label{fig:samplers}
\end{figure}

\subsection{Realistic Simulation}

In the realistic simulation, we exploit flexible conditional density estimators to fit an estimand-specific generating process illustrated in Fig.~\ref{fig:samplers} (right).
For example, the interaction of $(V_1, V_2)$ on $Y$ requires three density estimates, namely $\hat{P}_n(V_1|PCs)$, $\hat{P}_n(V_2|PCs)$ and $\hat{P}_n(Y|V_1, V_2, PCs, C)$.
Similarly, the single variant effect of $V_1$ on $Y$ requires $\hat{P}_n(V_1|PCs)$ and $\hat{P}_n(Y|V_1, PCs, C)$.
Once these conditional densities have been estimated, new data can be generated via ancestral sampling.
Throughout, the empirical distribution $\hat{P}_n(PCs, C)$ is used to sample from the root nodes.\\

For the simulation to be realistic, the conditional density estimators must be able to capture the complexity of the data, which has two main implications.
The first is that the causal model should include as many causal variables as possible to generate the corresponding descendent variables.
The second is that the density estimators must be flexible and data-adaptive to capture sufficiently complex data generating processes.
We next discuss in turn how we address both challenges.

\subsubsection{Variable Selection}

For each genetic variant $V_j$, we assume the set of parent variables consists of the first six principal components derived from the genotyping data.
While limitations of PCA have been pointed out~\citep{elhaik2022principal}, it is still standard practice to use PCA to adjust for genetic ancestry~\citep{price2006principal, mbatchou2021computationally, yang2011gcta, loh2015efficient}.\\

In principle, variation across the entire genome as well as environmental variables could be causal of a given trait's variation.
To restrict the dimensionality of the problem we only consider a small subset of these putative causes.
Environmental variables are kept to the standard covariates, namely, age at assessment and genetic sex.
To include potential causal variants from the genome, we use published GWAS results from GeneATLAS~\citep{canela2018atlas}.
More precisely, we select a maximum of 50 variants from all variants associated with the outcome of interest ($\text{p-value} < 10^{-5}$).
Furthermore, we require these variants to (i) be at least one million base pairs away from each other to avoid linkage disequilibrium (LD), and (ii) have a minor allele frequency of at least $0.01$.
In Fig.~\ref{fig:samplers}, these selected variants and environmental variables are contained within the $C$ variable.
The fact that selected variants are not generated from principal components greatly reduces computational burden but poses a mild limitation.
Since all variables in $C$ are jointly sampled, the dependence structure within selected variants is preserved.
However, selected variants and variants used as treatment to define the causal estimands, are independent in this simulation.
In a more realistic scenario they would only be independent once conditioned on principal components.

\subsubsection{Model Selection} \label{sec:sieve-neural-net}

The second requirement for simulations to be realistic is that the density estimators are sufficiently expressive to capture complex patterns, which means that the model class must be large.
Neural networks have been shown to be able to approximate a large class of functions, and to scale well to large datasets such as the UK Biobank~\citep{hornik1989multilayer,deepnull2022}.
We thus used two types of neural networks depending on the type of outcome variable.
For categorical variables, including binary outcomes, a multi-layer perceptron was used, while for continuous variables we used a mixture density network~\citep{haykin1998neural,bishop1994mixture}.
In all cases, in order to prevent overfitting, density estimators were trained as sieve estimators~\citep{chen2007large}.
That is, the size of the model was chosen data-adaptively by sequentially increasing the model capacity and applying an early-stopping criterion based on cross-validation performance~\citep{prechelt1998automatic}.
For illustration, a simplified training procedure is described in Algorithm~\ref{alg:sieve_network}.

\subsection{Considered Estimands} \label{sec:simulation-estimands}

We select a variety of estimation tasks representative of analyses in population genetics.
We focus on single-variant effects and epistatic effects which are quantified by the Average Treatment Effect (ATE) and Average Interaction Effect (AIE), respectively.

To account for the above challenges, we selected 25 estimands, listed in Table~\ref{table:simulation_estimands}, according to the following criteria.
All included estimands were supported by previous evidence of association to allow the realistic simulation to diverge from the null hypothesis of no effect.
Note however, that this is not mandatory and that the goal of the simulation is not the replication of these effect sizes.
For single variant effects, $5$ diverse traits were chosen:

\begin{itemize}
    \item Leukocyte count: Count variable.
    \item Body mass index (BMI): Continuous variable.
    \item Sarcoidosis (self-reported): Binary variable, $\approx 1000$ cases (rare).
    \item Multiple sclerosis: Binary variable, $\approx 1900$ cases.
    \item Other diseases of the digestive system (K90-K93): Binary variable, $\approx 25000$ cases.
\end{itemize}

For each trait, two variants were manually selected based on previous GWAS results from the GeneATLAS~\citep{canela2018atlas}, requiring $\text{p-value} < 10^{-5}$. 
These two variants were not in linkage disequilibrium and had different minor allele frequencies.\\

For interaction analyses, we used the following previous results from the literature:

\begin{itemize}
    \item Skin colour, shown to be associated with (rs1805007, rs6088372), (rs1805005, rs6059655) and  (rs1805008, rs1129038)~\citep{morgan2018genome}.
    \item Parkinson's disease, shown to be associated with (rs1732170, rs456998, rs356219, rs8111699) and (rs11868112, rs6456121, rs356219)~\citep{fernandez2019snca}
    \item Multiple sclerosis, shown to be associated with (rs10419224, rs59103106)~\citep{singhal2023evidence}
    \item Psoriasis, shown to be associated with (rs974766, rs10132320)~\citep{singhal2023evidence}
\end{itemize}

These estimands represent both rare and frequent genetic variations as well as a variety of trait types including continuous, count and binary; see Table~\ref{table:simulation_estimands} for full details.
Together, these estimands provide a representative set of analyses to assess the performance of semi-parametric estimators in population genetics.

\subsection{The Estimators} \label{sec:simulation-estimators}

We investigate the performance of the OSE and the wTMLE in both their canonical and cross-validated versions.
For the estimation of nuisance functions, we consider two strategies: (i) a GLMNet-based approach~\citep{friedman2010regularization} incorporating all 2-point interaction terms involving genetic variants, (ii) an XGBoost-based approach~\citep{chen2016xgboost} where the $\verb|(max_depth, lambda)|$ hyper-parameters are selected via cross-validation using a simple grid search.
In all cases, the cross-validation scheme is a 3-fold cross-validation, stratified across the variants acting as treatment variables, and the outcome variables.

\subsection{Results} \label{sec:simulation_results}

To investigate the performance of semi-parametric methods on biobank-scale data, we perform the analysis across two different dataset sizes ($50\,000, 500\,000$).
A simulation task is thus a triple: sample size, estimator, estimand.
For each task, a grid of $500$ bootstrap samples was run on the University of Edinburgh high-performance computing \href{https://www.ed.ac.uk/information-services/research-support/research-computing/ecdf/high-performance-computing}{Eddie} cluster.

\subsubsection{Null Simulation}

Due to the absence of confounders in the null simulation of Fig.~\ref{fig:samplers} (left), model misspecification is not expected to affect the estimation of the (interaction) effect of genetic variants on outcome.
All proposed estimators should thus provide asymptotic coverage at the nominal confidence level across all estimates (\eg, $95\%$).
Fig.~\ref{fig:null-simulation}A shows that this is the case for most, but not all, estimands.
For all four estimators (CV-)OSE and (CV-)wTMLE, there is at least one estimand for which the coverage is below $50\%$.
This could be because the asymptotic regime has not been reached, possibly due to near-positivity violations for rare variants, noting that coverage is improved in the larger simulation.
Exact control of positivity is challenging since it depends on the true, but unknown value of $p(V=v|W=w)$ for all observed $(v,w)$.
In the null simulation however, $p(V|W) = p(V)$ since $V$ and $W$ are sampled independently.
We assess the impact of positivity by excluding estimands that fail to meet a specified threshold in the original UK Biobank dataset, \eg, $p(V=v) < 0.01$.
Since estimands are multi-dimensional, representing various genotype changes, rather than discarding an entire estimand, we exclude only the genotype changes that do not meet the given threshold.
This approach is more conservative than the data-adaptive selection criterion for the propensity score truncation level proposed in~\citep{10.1093/aje/kwac087}.
\begin{figure}[H]
    \centering
    \includegraphics[scale=0.51]{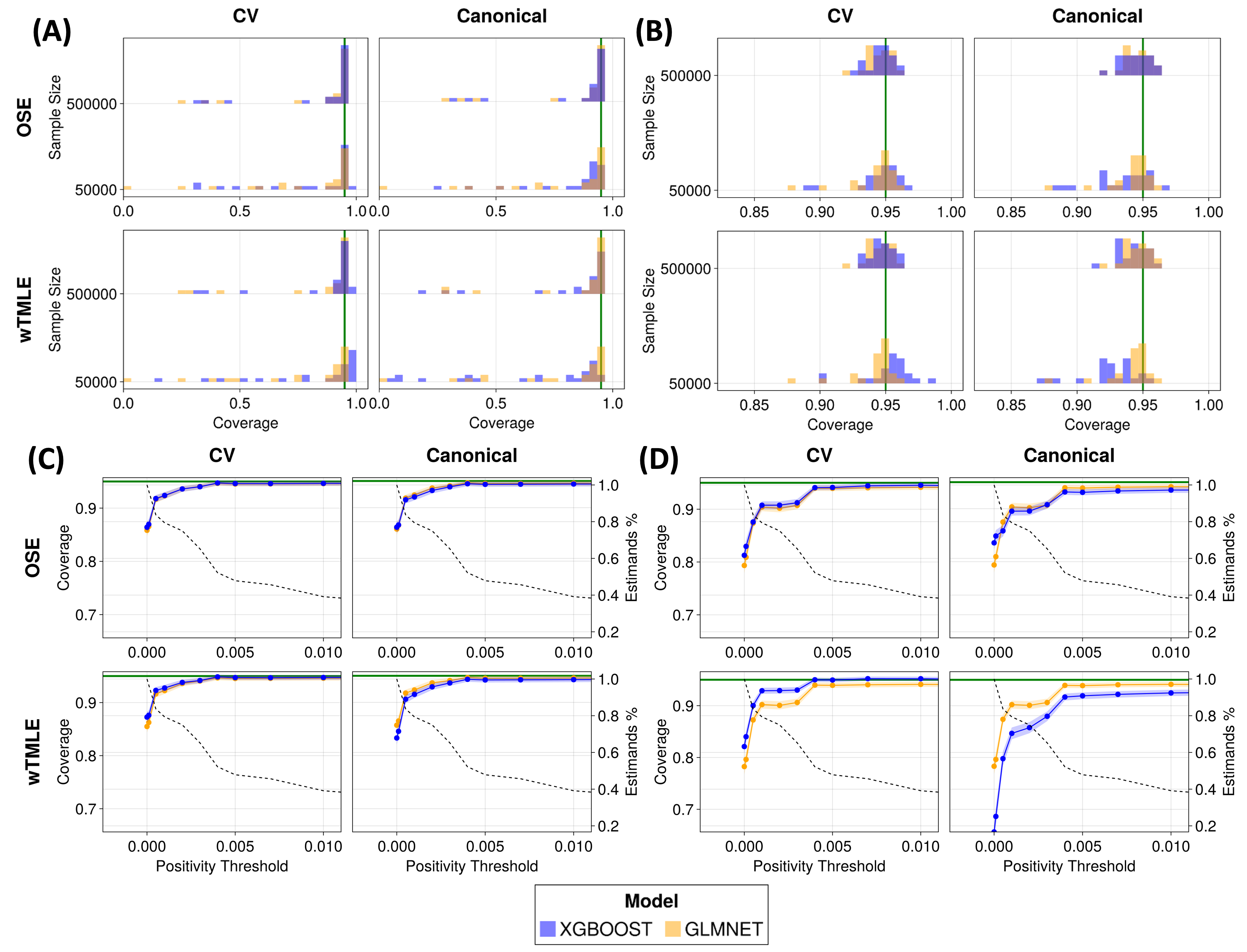}
    \caption{\textbf{Null Simulation Coverage.} Each plot is divided in 4 quadrants and presents coverage results. Rows represent the estimators (OSE, wTMLE) and columns the resampling strategies (CV, Canonical). 
    \textbf{(A-B) Coverage distributions across estimates when (A) no positivity threshold is imposed and, (B) a positivity threshold of $\geq 0.01$ is imposed.} Each quadrant is further sub-divided by sample size, $500\, 000$ (top) and $50\,000$ (bottom). All estimators reach nominal coverage (green line) across most, but not all, estimands when no positivity threshold is imposed. Larger sample sizes yield better results and, when the positivity threshold of $\geq 0.01$ is imposed, essentially nominal coverage for all four estimators (B). This indicates that the asymptotic regime has not been reached for all estimands at $50\, 000$ samples. As the positivity threshold increases, the coverage reaches its nominal level as demonstrated in panels C and D.
    \textbf{(C-D) Mean Coverage Across Positivity Thresholds ($\mathbf{500\, 000}$ (C) and $\mathbf{50\,000}$ (D)).} A point corresponds to a positivity threshold (x-axis) and mean coverage across estimands (y-axis). Only estimands meeting the positivity threshold are included in the mean coverage computation. The black decreasing dotted line represents the percentage of total remaining estimands' components after imposing the positivity threshold. For readability, the plots are limited to  positivity thresholds $< 0.01$. Above this threshold, the nominal coverage was reached across all estimates.}
    \label{fig:null-simulation}
\end{figure}
In Fig~\ref{fig:null-simulation}B, coverage estimates are shown only for those estimands which satisfy the threshold $p(V=v) \geq 0.01$, and coverage is seen to be near nominal.
Furthermore, by varying the positivity threshold value of $0.01$, we analyse the estimators' sensitivity to practical positivity violations.
Fig.~\ref{fig:null-simulation}C-D present coverage results across various positivity thresholds for sample sizes $500\,000$ (C) and $50\,000$ (D) respectively. 
As anticipated, convergence to nominal coverage is faster for the larger sample size.
However, for both sample sizes, on average, nominal coverage is achieved for all genotypes with frequencies as low as $0.005$.
This holds true across all estimates, except for the canonical wTMLE using XGBoost when the sample size is $50\,000$.
The likely cause of this deviation is overfitting, as the cross-validated estimator restores coverage to the desired level.
The percentage of remaining estimands drops rapidly with increasing positivity thresholds.
For the conservative $0.01$ threshold, approximately $40\%$ of estimands' components remain. \\

While of primary importance to researchers, coverage provides an incomplete view of an estimator's performance.
We additionally analyse the bias-variance decomposition across all estimation tasks.
Squared bias, variance and mean-squared error are estimated via bootstrap resampling as per Section~\ref{sec:bootstrap_estimators} and provided in Supplementary Table 6.
Here, the variance is not estimated using the influence function but as the sample variance of estimates on bootstrapped data.
However, these are comparable provided a sufficient positivity constraint is applied (see Fig~\ref{fig:appendix-bootstrap-vs-infcurve-var}).
For direct comparison, the variance of the influence function is also provided in Supplementary Table 6. For the null simulation, these are almost always similar (less than $15\%$ relative difference) for both sample sizes when using GLMNet, and for the larger sample size $500\,000$ when using XGBoost.
As expected, the bias and variance of estimators decrease with sample size for almost all estimands. The only $4$ exceptions (out of $100$), occur for the rarest trait in our dataset (sarcoidosis).\\

Interestingly, when using XGBoost, the bias of cross-validated estimators is always larger than that of their canonical counterpart.
This unexpected result seems to indicate that the asymptotic regime is not reached, possibly due to the low number of folds used throughout ($k=3$).
Effectively, each nuisance function is fitted on $2/3$ of the data, increasing the risk of overfitting on each partition.
In principle, this limitation could be overcome by increasing the number of folds.
However, the increased computational complexity is challenging in population studies of this scale.
Inspecting results from the same analysis when using GLMNet, a model less prone to overfitting, for nuisance function estimation, shows that bias estimates are more varied between the canonical and cross-validated estimators.
This observation supports the potential explanation for higher bias in cross-validated estimators when using a more flexible algorithm such as XGBoost for nuisance function estimation. \\

Together, these results demonstrate that in the absence of genetic effects, the false discovery rate can be appropriately controlled provided genetic variations are as frequent as $0.01$.
The cross-validated estimators implemented in TarGene can be used to trade bias and variance for improved coverage.

\subsubsection{Realistic Simulation} \label{sec:realistic-simulation}

We now turn to the results of the realistic simulation, where genetic variants influence traits, potentially in complex ways, as modelled by the Sieve Neural Network Estimator (SNNE) Algorithm~\ref{alg:sieve_network}.
The empirical performance of SNNE across different densities is shown in Supplementary Figure~\ref{fig:density-estimates-comparison}, demonstrating that this approach offers a flexible, data-adaptive solution for modelling complex data-generating processes while mitigating the risk of overfitting.\\

The issue of positivity, discussed in the previous section, persists in the realistic simulation, as shown in Fig.~\ref{fig:realistic-simulation}A, and is similarly resolved when imposing a positivity threshold $p(V=v) \geq 0.01$ as shown in Fig.~\ref{fig:realistic-simulation}B.
Fig.~\ref{fig:realistic-simulation}C-D indicate that for the same positivity threshold of $0.005$, nominal coverage is achieved when cross-validated estimators are combined with XGBoost.
This is likely due to two factors: (i) XGBoost effectively captures the complexity of the true data-generating process, and (ii) while XGBoost tends to overfit, cross-validation compensates for this tendency.
In large datasets, XGBoost consistently provides better coverage than GLMNet, regardless of the resampling strategy.
In settings such as the UK Biobank, the canonical version of XGBoost-based estimators offers an interesting trade-off between computational efficiency and coverage, with only a marginal $\approx1\%$ drop.
However, in smaller datasets, XGBoost underperforms when cross-validation is not employed, likely due to overfitting.
The analysis of bias, variance and mean squared error, similar to that of the null simulation, mostly leads to the same conclusion. The improved coverage of cross-validated estimators is largely gained at the expense of bias and variance when models prone to overfitting are used.
The sample variance and variance of the influence function are more likely to differ (i.e., more than $15\%$ relative difference) at the smaller sample size of $50\,000$ when using XGBoost; the difference is comparable between sample sizes when using GLMNet. This suggest model misspecification plays a larger role in the realistic simulation than in the null simulation where treatment is, effectively, randomly assigned, as expected. \\

\begin{figure}%[H]
    \centering
    \includegraphics[scale=0.52]{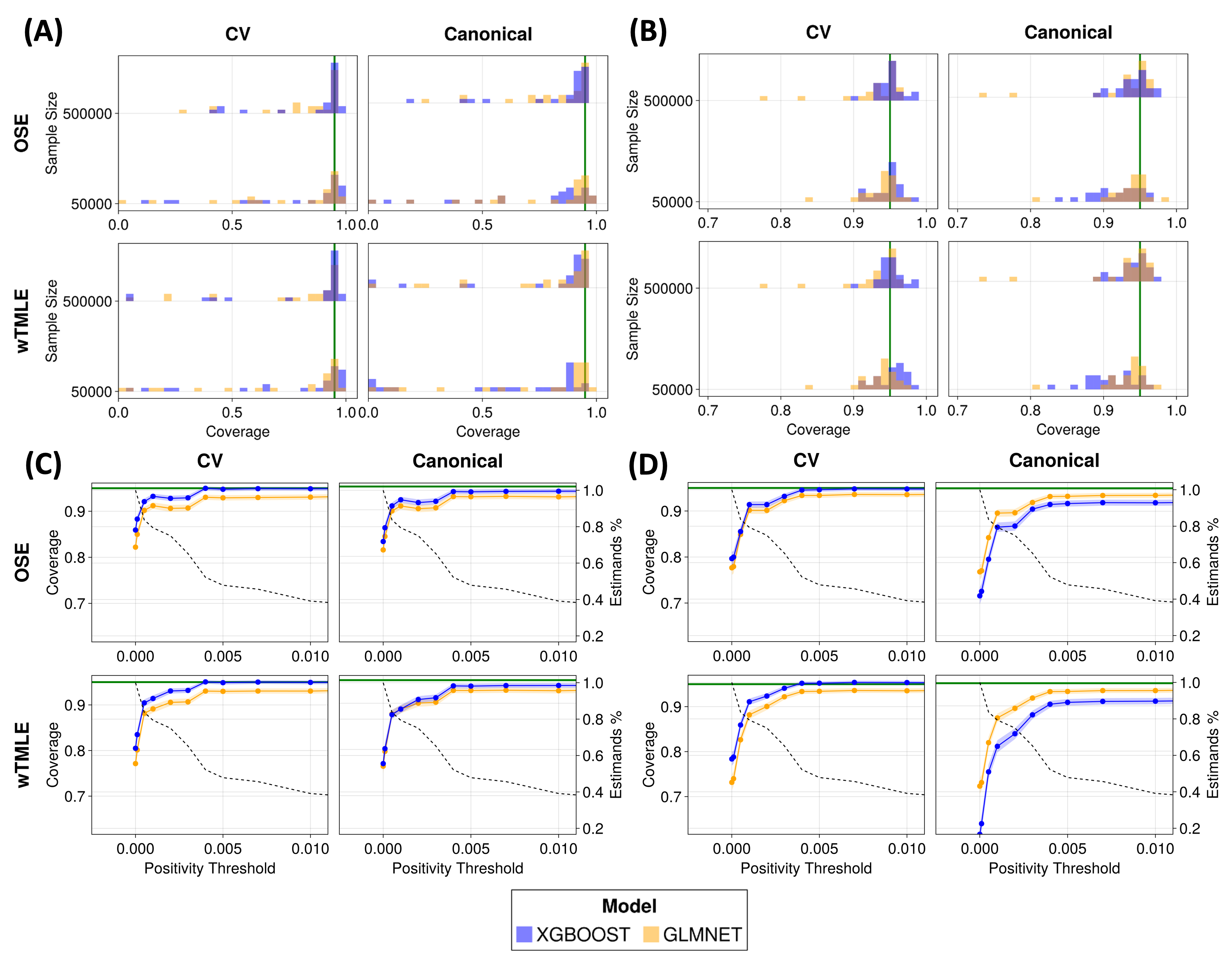}
    \caption{\textbf{Realistic Simulation Coverage.} The panel is organised exactly as Fig.~\ref{fig:null-simulation}. \textbf{(A-B) Coverage distributions across estimates when (A) no positivity threshold is imposed and, (B) a positivity threshold of $\geq 0.01$ is imposed.} Similarly to Fig.~\ref{fig:null-simulation}, larger sample sizes yield better results, more so for the smaller sample size of $50\, 000$ indicating that the asymptotic regime has not been reached for all estimands. As the positivity threshold increases, the coverage reaches its nominal level as demonstrated in panels B and C with the exception of the canonical OSE and wTMLE at sample size $50\, 000$. \textbf{(C-D) Mean Coverage Across Positivity Thresholds ($\mathbf{500\, 000}$ (C) and $\mathbf{50\,000}$ (D)).} In large sample settings, XGBoost-based estimators perform better than their GLMNet counterpart. For smaller sample-sizes, cross-validation is preferred for these models as they are more prone to overfitting.}
    \label{fig:realistic-simulation}
\end{figure}

In tandem with nominal coverage, we present the power of wTMLE across estimands, sample sizes, resampling schemes and models used to estimate nuisance functions in Fig.~\ref{fig:power-wtmle}.
A similar plot is presented for the OSE in Supplementary Figure~\ref{fig:power-ose}.
As expected, increased sample sizes result in higher power for both models.
Interestingly, cross-validation affects the power of the XGBoost model but not the GLMNet model.
This suggests that in the absence of overfitting there is no or little loss of power from cross-validated estimators.
Furthermore, this analysis shows that the power to detect ATEs (right of the dashed lines) is larger than the power to detect AIEs (left of the dashed lines).
This is expected since the variance, and hence power, of an estimator depends on the complexity of the estimand.

\begin{figure}
    \centering
    \includegraphics[scale=0.3]{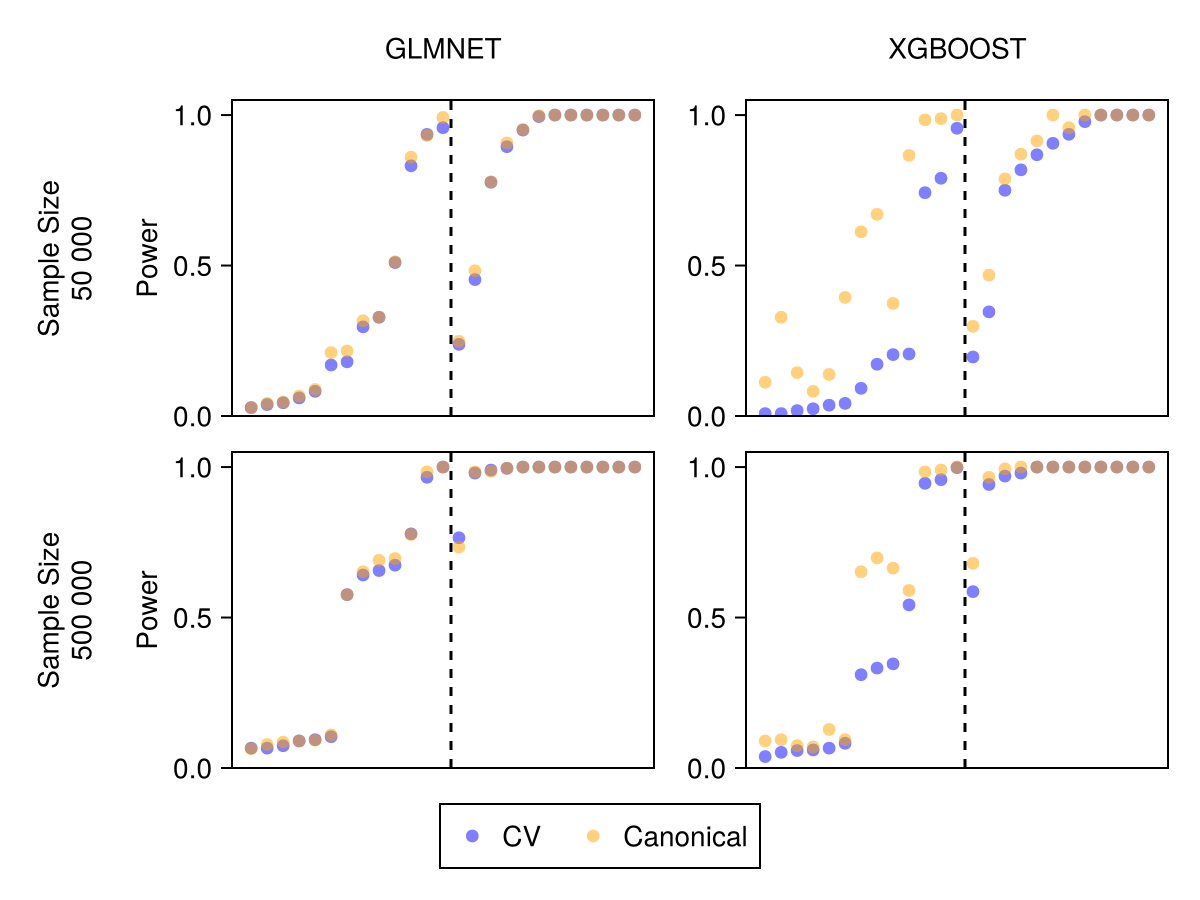}
    \caption{\textbf{Power analysis of weighted TMLE.} The plot is organised in 4 quadrants. Rows indicate sample sizes with $n = 50\, 000$ (top) and $n = 500\, 000$ (bottom), columns indicate the model used to fit nuisance functions with GLMNet (left) and XGBoost (right), and colour indicates the resampling scheme, \ie, cross-validated (blue) and canonical (orange). Each dot corresponds to a single estimand and the dashed lines separate AIEs (left) and ATEs (right).}
    \label{fig:power-wtmle}
\end{figure}

%% file: applications.tex
\section{Application to statistical genetics and large biobanks}\label{sec:UKB}

We performed a series of analyses using (i) the UK Biobank (UKB) cohort~\citep{UKB_bycroft}, and (ii) the All of Us (AoU) cohort~\citep{AllOfUs} to demonstrate the performance of semi-parametric estimators in the context of real-world genotype-phenotype inference.
For comparison with previous studies such as the GeneATLAS~\citep{canela2018atlas}, all analyses presented here were performed on the subpopulation of individuals with a self-reported white ethnic background.
Specifically, in UKB we selected individuals with data field $21000$ subcoding in the category $1$ (``White"), namely $1001$ (``British"), $1002$ (``Irish"), and $1003$ (``Any other white background").
In the AoU cohort, we included only individuals from non-Hispanic or Latino descent and inferred European genetic ancestry.\\

First, we contrast our approach with the field's commonly employed Linear Mixed Model method by performing a phenome-wide association study on UKB for a well studied variant in the \textit{FTO} gene region: rs1421085.
We compare the estimated effect of this SNP on BMI in the UKB with the corresponding estimate in AoU.
Second, we reveal gene-by-environment interactions between rs1421085 and two deprivation indices.
Third, we replicate pairwise genetic interactions previously reported for hair colour by~\cite{morganGenomewideStudyHair2018} and report additional evidence of interactions for both skin and hair colour.
And fourth, we show how TarGene can investigate higher-order interactions using multiple loci related to vitamin D receptor (VDR) function.

\subsection{Setting}

For all analyses we used the three canonical estimators, \ie, TMLE, wTMLE and OSE.
The nuisance functions $Q$ and $g$ were estimated using the practical recommendations for Super Learning~\citep{phillips2023practical}.
In brief, we use (i) $k$-fold cross-validation or stratified $k$-fold cross-validation based on the outcome type (continuous or binary, respectively), here $3 \leq k \leq 20$, selected adaptively based on the rarest class of each outcome, and (ii) included the constant fit, a regularized logistic/linear regression (ridge, $\lambda = 1$), a gradient-boosted tree ($\verb|n_round| = 100$, default parameters otherwise), and HAL with hyper-parameters \verb|max_degree| $=1$, \verb|smoothness_orders| $=1$, \verb|lambda| $= 30$~\citep{7796956}, as base learners.
However, we note that for the optimal performance of HAL in more bespoke analyses, the parameter $\lambda$, bounding the total variation norm, should be left unspecified so that it is chosen by the algorithm's internal cross-validation.\\

For confounding adjustment, we used the first 6 PCs computed from the genotypes.
We also added age at assessment and genetic sex as explanatory variables for the outcome model $\hat{Q}$. 
To investigate the impact of Sieve Variance Plateau correction, variance estimates were corrected using $100$ genetic similarity thresholds $\tau$.
Full run configurations details can be found \href{https://github.com/edbiomedai/targene-paper/tree/main}{\textcolor{blue}{online}}.

\subsection{\textit{FTO} PheWAS}

We conducted a phenome-wide association study (PheWAS) using UKB data for $768$ traits.
We compared our results with those from the LMM analysis of GeneATLAS by~\citep{canela2018atlas}.
For a comparison across populations, we contrast our UKB results on BMI obtained by TarGene with the corresponding results on BMI when running TarGene on the AoU cohort.
We chose a well-studied variant, rs1421085, located in the first intron of the \textit{FTO} gene. For this variant, the T-to-C nucleotide substitution has been predicted to disrupt the repression of {\it IRX3} and/or {\it IRX5}, thereby leading to a developmental shift from adipocyte browning to whitening programs and loss of mitochondrial thermogenesis~\citep{claussnitzerFTOObesityVariant2015a}.
This variant has also been associated with several related traits, such as BMI and obesity~\citep{fraylingCommonVariantFTO2007}.
Since rs1421085 has two different alleles, the three possible genotype changes we estimate for each trait are (i) TT~$\to$~TC, (ii) TC~$\to$~CC, and (iii) TT~$\to$~CC. \\

Since TarGene relies on semi-parametric estimation theory, we can directly assess whether the effect on outcome of the first C substitution (TT~$\to$~TC) is different from the effect on outcome of the second substitution (TC~$\to$~CC) by an application of the functional delta method~\citep{van2000asymptotic}.
This difference can be estimated with the semi-parametric estimator $\hat{\psi}_{\Delta,n}$ of Eq.~\ref{eq:nonlinear_allelic_difference}.
While such questions can be addressed with LMM methods, \eg, by additive-dominance models~\citep{doi:10.1126/science.abn8455}, typically large-scale LMM studies neglect allelic effect differences, instead assuming that the two genotype changes have the same effect.
To allow for a comparison of TarGene results with LMM methods, we compare GeneATLAS results to our TT~$\to$~TC estimate, noting that they are comparable if the Allelic Effect Difference $\psi_{\Delta}$ is zero.

\subsubsection{Comparing Effect Sizes between methods}

We begin by illustrating the difference between our reported TT~$\to$~TC effect on BMI, and the effects present in GWAS catalogues~\citep{buniello2019nhgri} in Fig.~\ref{fig:rs1421085-plots}A. 
We discuss allelic effect differences more generally in the next section.
The figure shows that all three double-robust estimators are concordant and report statistically lower effect sizes than the two linear/linear-mixed models which, in contrast, are discordant with one another.
Apart from BMI, semi-parametric point estimates are mostly aligned with those produced by LMMs (Fig.~\ref{fig:rs1421085-plots}B, left).
This is expected since in the absence of confounding, effect estimates with linear models are robust to model misspecification, offering reliable estimates even when the true relationship between variables is not perfectly linear.
For rs1421085, the absence of confounding by principal components is shown in Fig.~\ref{fig:appendix-rs1421085-UKB-pcs}.
Therefore, for this variant, the main source of differences between semi-parametric and linear estimates will likely be due to allelic effect differences. \\

Furthermore, Fig.~\ref{fig:rs1421085-plots}B shows that the distribution of p-values is shifted towards less significant values as compared to GeneATLAS estimates.
After FDR correction across the $768$ traits, TarGene finds fewer significant results at the $0.05$ level ($63$ traits) than GeneATLAS ($159$ traits).
This phenomenon highlights the risk of an increased false discovery rate due to overly restrictive parametric assumptions.
A summary table of all significant estimation results after multiple testing adjustment is provided in Supplementary Table~$1$.

\begin{figure}%[H]
    \centering
    \includegraphics[scale=0.69]{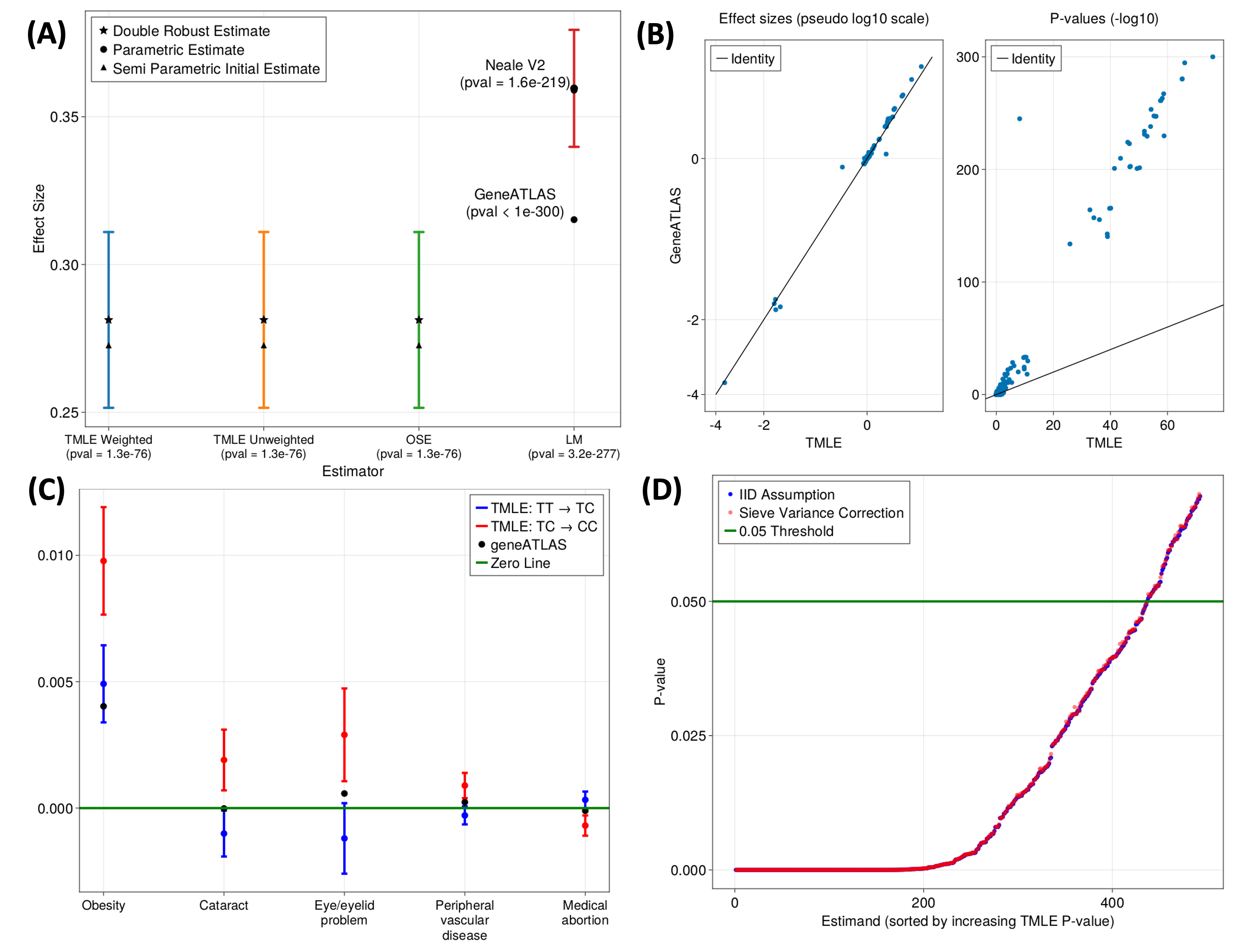}
    \caption{\textbf{TarGene results and comparison with the linear mixed model on UKB data.}
    ({\bf A})~\textbf{Inference results}. Comparison of methods to estimate the effect size of rs1421085 on body mass index (BMI; UK Biobank Data-Field 23104). All double robust estimators share the same initial fit and apply different targeting strategies: weighted TMLE (blue), unweighted TMLE (orange), and OSE (green). The three estimates are concordant and exhibit a statistically significant lower effect size than the linear model based inference (red). Neale V2 and GeneATLAS use a linear model and linear-mixed model, respectively, but do not report standard deviations.
    We refit the Neale V2 linear model on this data to obtain a confidence interval (red).
    The central value of GeneATLAS does not lie within the 95\% confidence interval of Neale V2.
    In contrast, all three double robust estimators are in complete agreement.
    See Panel C of Fig.~\ref{fig:appendix-rs1421085-plots} for a comparison between the initial estimates of effect sizes, relying on Super Learning, against the final effect sizes after the TMLE step.
    ({\bf B})~\textbf{Comparison with GeneATLAS}. Comparison of effect sizes (left) and p-values (right) reported by TarGene (TMLE) and GeneATLAS (LMM). Effect sizes are concordant overall on this study but the p-values reported by TarGene are more conservative. Because TarGene's p-values are more robust to model misspecification, GeneATLAS' results likely contain an inflated set of false discoveries.
    ({\bf C}) \textbf{Non-Linear effects}. A selection of traits for which rs1421085 TT~$\to$~TC and TC~$\to$~CC effect estimates are significantly different; Supplementary Table~2 contains the complete list. Effect sizes are reported with associated $95\%$ confidence intervals together with estimates from GeneATLAS' LMM fits (black data points)~\citep{canela2018atlas}. The latter almost always fall in-between our TT~$\to$~TC and TC~$\to$~CC estimates, indicative of an averaging effect. [Continued]
    }
\label{fig:rs1421085-plots}
\end{figure}

\begin{figure}[t]
\contcaption{
    ({\bf D}) \textbf{Sieve variance correction}. P-values obtained from two variance estimation methods for rs1421085. In red, the individuals in the UK Biobank are assumed to be independent and identically distributed (iid), while in blue, a sieve correction method is applied to account for the population dependence structure. Each p-value corresponds to a specific parameter of interest for which the initial iid estimate was under the 0.07 threshold.
    See Panels A and B of Fig.~\ref{fig:appendix-rs1421085-plots} for the unfiltered plot and an example of a Sieve Plateau curve, and Panel D for the histogram of the genetic relationship matrix.
    }
\end{figure}

\subsubsection{Allelic Effect Differences} \label{sec:fto_ukb}

We find 39 traits for which rs1421085 displays a significant Allelic Effect Difference, 35 of which are highly correlated with BMI. For instance, we find that the departure from TT to TC is associated with an increased weight (UKB field 21002) of $0.78$~kg ($95\%$ CI: $0.69 - 0.86$~kg). In comparison, the departure from TC to CC is associated with a significantly larger increase of $1.33$~kg ($95\%$ CI: $1.20 -  1.45$~kg). For illustration, a subset of significant non-linear traits is presented in Fig.~\ref{fig:rs1421085-plots}C; Supplementary Table~$2$ contains the complete list. As might be expected, most estimates reported by GeneATLAS, based on a LMM approach, fall in-between estimates from our two scenarios, representative of an averaging effect.\\

Notably, some traits seem to display opposite effect sizes (before multiple testing adjustment) for the two allelic changes TT~$\to$~TC and TC~$\to$~CC.
Thus TarGene can capture variant-trait pairs displaying heterozygote advantage~\citep{Heterozygote_Advantage}. Such patterns cannot be detected by a linear model that assumes equal allelic effect sizes.

\subsubsection{Sieve Plateau Variance Estimation} \label{sec:svp_practical}

In Section~\ref{sec:svp}, we proposed to account for the dependence among individuals in the variance estimates of semi-parametric estimators using Sieve Plateau (SP) variance estimation of~\citep{SievePlateau}.
The SP method is computationally intensive, requiring (i) computation of the GRM, an $n \times n$ matrix (for UKB, $n \approx 450\,000$), and (ii) for each estimand and each threshold $\tau$, matrix multiplications involving GRM components and influence curves.
However, since SP only increases variance estimates, it is sufficient to consider estimates that are significant at a given threshold (\eg, $p < 0.05$) before SP variance correction.
In Fig.~\ref{fig:rs1421085-plots}D, we show the effect of this SP variance correction for all effect sizes obtained for rs1421085 (with initial $p < 0.07$ for improved readibility).
The p-values resulting from both the iid (red) and the sieve variance plateau estimators (blue) are reported and show little difference.
An example sieve plateau variance curve for the effect estimate of rs1421085 on BMI shows an increased variance of $\approx 1.4\%$ (Fig.~\ref{fig:appendix-rs1421085-plots}).\\

The covariance term of influence functions of related individuals are thus likely negligible as compared to the individual variance terms in Eq.~\ref{eq:SP_variance_estimator}.
The impact of the SP variance method in a more diverse population, or in more related family-based cohort, is an interesting research direction.

\subsubsection{Comparing effect sizes between cohorts}

Next, we sought to investigate whether the UKB results in Section~\ref{sec:fto_ukb} replicate in another large biobank, noting potential differences in allelic frequencies and environmental factors between biobanks.
To explore this, we leveraged the All of Us (AoU) cohort~\citep{AllOfUs}, a United States-based cohort.
This study was performed on AoU's cloud-based platform called the Researcher Workbench, which is compatible with the TarGene software.\\ 

We first constructed a cohort with inclusion/exclusion criteria that matched our UKB study, requiring that each participant had complete data available for BMI, age at BMI measurement, sex at birth, and genetics.
Further, we included only individuals of inferred European genetic ancestry and of non-Hispanic or Latino origin to mirror inclusion criteria for our UKB study.
Principal components were computed across all remaining participants ($n = 122,752$) and genotyped SNPs, excluding any variants in LD with our SNP-of-interest, rs1421085.
In this cohort, 6 PCs were sufficient to capture population stratification and were used as confounding variables for the estimation step, see Fig. \ref{fig:appendix-rs1421085-AoU-screeplot}.\\ 

We found that rs1421085 also had a significant effect on BMI for each allelic change, consistent with results found in the UKB.
While TT to TC was associated with an increase of $0.59~kg/m^2$ ($95\%$ CI: $0.42 - 0.77$) and p-value $< 1\times10^{-10}$, the departure from TC to CC was associated with an increase of $0.88$~kg/m${}^2$ ($95\%$ CI: $0.47 - 1.28$) and p-value $< 1\times10^{-4}$.
Estimates found in the AoU cohort were larger than those found in UKB, but trends were consistent between cohorts.
The effect of the allele change $TC \rightarrow CC$ on BMI was also shifted higher than $TT \rightarrow TC$, however, estimates in the AoU cohort harboured higher levels of uncertainty, and so no significant non-linear effect was found between each individual allele change in AoU (p-value $= 0.27$), see Fig.~\ref{fig:appendix-rs1421085-BMI-AoU-UKB-effects}.
The increased uncertainty may be due to the smaller sample size of our AoU cohort ($n = 122,\!752$) relative to the UKB cohort ($n = 459,\!207)$, a decrease in sample size (to $\approx 27\%$ of UKB) and minor allele frequency for the C allele ($\approx 10\%$ lower) in AoU, resulting in substantially reduced sample sizes for the CC genotype, see Fig. \ref{fig:appendix-rs1421085-BMI-AoU-UKB-effects}.
Alternatively, this effect may also be due to increased variation of BMI in the AoU cohort due to environmental effects.

\subsection{Gene-by-environment interaction}

In the previous section, our TarGene PheWAS on UKB confirmed that rs1421085 is significantly associated with BMI and various BMI-related traits. This association between rs1421085 and BMI was also found to be replicated in the All of Us cohort.
Since BMI has also been associated with area-based deprivation \citep{Twaits_Alwan}, it is natural to investigate the potential interactions between rs1421085 and deprivation.
There are currently two main measures of deprivation in the UK: the Townsend Deprivation Index (TDI) and the Index of Multiple Deprivation (IMD) used by~\citep{Twaits_Alwan}.
We used these indices in two separate phenome-wide interaction studies (PheWIS), one between rs1421085 and TDI and one between rs1421085 and IMD.
Since deprivation indices are continuous quantities, we discretised them using quintiles and compared the most extreme quintiles.
For rs1421085, we compare the three genotype groups TT, TC, CC.\\

We found 21 significant BMI related traits after FDR correction ($0.05$ level) that are captured by both TDI and IMD (Supplementary Table~$3$). For instance, whilst we have seen that an increase in the number of C alleles in an individual is associated with an increase in body weight and that most deprived individuals are more likely to be overweight, the interaction of these factors is super additive: an increase of $1.07$~kg (p-value: $1.09\times10^{6}$, adjusted p-value: $1.4\times10^{-3}$) for TDI, and an increase of $0.91$~kg (p-value: $3.89\times10^{-5}$, adjusted p-value: $8.0\times10^{-3}$) for IMD.

\subsection{Epistatic interaction}

Detection of epistasis in complex traits is challenging~\citep{Haley_Epistasis,Mackay2024}. For example, epistatic interactions are expected to be much smaller than main effect sizes, which can already be small for polygenic traits. In this section, we explore the potential for semi-parametric estimators to reveal such interactions. For that purpose, we rely on a study investigating hair colour, in which nine pairs of variants were reported to be statistically interacting with red-hair using a logistic regression model and a likelihood ratio test~\citep{morganGenomewideStudyHair2018}. We note, however, that the likelihood ratio test statistic measures interactions on a multiplicative scale while we investigate interactions on an additive scale which is often of more direct public health relevance \citep{VanderWeele_interactions}. In particular, the existence of interactions on one scale does not imply the existence of interactions on the other scale. We found that five of the nine reported epistatic results (Table~\ref{table:hair-color}) are also revealed by semi-parametric estimation methods. Two were not reproduced and two were not computed because they did not pass the marginal positivity threshold.
   
\begin{table}[H]
\centering
 \begin{tabular}{||c c c c||} 
 \hline
 Variant 1 & Variant 2 & Effect-size & P-value \\ [0.5ex] 
 \hline\hline
 rs1805005 (GG~$\to$GT) & rs6059655 (GG~$\to$AG) & $1.8\times10^{-2}$ & $1.0\times10^{-19}$ \\ 
 rs1805007 (CC~$\to$CT) & rs6088372 (CC~$\to$CT) & $3.0\times10^{-2}$ & $1.4\times10^{-40}$ \\
 rs1805008 (CC~$\to$CT) & rs1129038 (TT~$\to$CT) & $-1.9\times10^{-2}$ & $2.9\times10^{-15}$ \\
 rs2228479 (GG~$\to$GA) & rs6059655 (GG~$\to$AG) & $-1.6\times10^{-2}$ & $1.5\times10^{-24}$ \\
 rs885479 (GG~$\to$GA) & rs6059655 (GG~$\to$AG) & $-1.5\times10^{-2}$ & $1.6\times10^{-16}$ \\ [1ex] 
 \hline
 \end{tabular}
 \caption{Summary table of reproduced significant results for red hair color.}
\label{table:hair-color}
\end{table}

In total, we find 27 significant epistatic signals for traits corresponding to either skin or hair colour (Supplementary Table~$4$).
This is expected because hair and skin colour are known to co-vary~\citep{sulemGeneticDeterminantsHair2007}.

%% file: software.tex
\section{Software and pipeline}\label{sec:Software}

As part of this work, we release two open-source packages and one software for scalable semi-parametric estimation of causal effects.

\subsection{General Purpose Package and CLI}

\href{https://targene.github.io/TMLE.jl/stable/}{TMLE.jl} is a \href{https://julialang.org/}{\texttt{Julia}} package for the semi-parametric estimation of causal effects from tabular datasets. The package currently supports one-step and 
targeted minimum loss-based estimation of the counterfactual mean, the average treatment effect and the average interaction effect.
The package also supports multi-dimensional and categorical treatment variables.
The associated \href{https://github.com/TARGENE/TMLECLI.jl}{TMLECLI.jl} provides an executable command-line interface.

\subsection{Population Genetics Nextflow Pipeline}

The \href{https://targene.github.io/targene-pipeline/stable/}{TarGene} software is a scalable \href{https://www.nextflow.io/}{\texttt{Nextflow}} pipeline for semi-parametric estimation of genetic effects from population cohorts.
We show below that computationally intensive PheWAS and GWAS studies are possible on modern computing resources.
In our case, all runs were performed on the Edinburgh high-performance \href{https://www.ed.ac.uk/information-services/research-support/research-computing/ecdf/high-performance-computing}{Eddie} cluster.
While TarGene comprises standard pre-processing procedures, we focus on the two unique processes: (i) the semi-parametric estimation process and (ii), the sieve variance plateau estimation process.

\subsubsection{Semi-Parametric Estimation Process}

In this section, we investigate the run time of the TMLE process for the two most common genetic studies: GWAS and PheWAS.
In both cases, we are thus computing the Average Treatment Effect for each individual variant on trait by comparing the major/minor to the major/major genotype.
Covariates were set to include the first 6 principal components, age and sex.
The benchmark is performed on a single core compute node.\\

We investigate the following four nuisance parameters estimation strategies (applied to both the outcome regression $\hat{Q}$ and the propensity score $\hat{g}$) from the most basic to the most comprehensive:

\begin{itemize}
    \item GLM: Standard generalized linear model
    \item GLMNet: GLM with regularization hyperparameter tuning over 3-folds cross-validation.
    \item XGBoost: The \href{https://xgboost.readthedocs.io/en/stable/}{gradient boosting trees} method with hyper-parameter tuning over 10 different settings in a 3-folds cross-validation scheme.
    \item SL: Super Learning including both the previous XGBoost and GLMNet combined with an outer 3-folds cross-validation.
\end{itemize}

We first focus on the PheWAS setting for which run time estimates are provided in Table~\ref{tab:phewas}.
In a PheWAS the estimation of the propensity score, $\hat{g}$, only needs to be performed once and can be re-used across all traits.
The computational complexity is thus driven by the estimation of each regression, $\hat{Q}$, and associated targeting steps.
The same remark holds for the targeting steps corresponding to the various genetic changes.
Computing the effects of the additional major/minor~$\to$~minor/minor and major/major~$\to$~minor/minor only costs two additional targeting steps while re-using the current $\hat{Q}$.
In all cases, running a PheWAS using TarGene is feasible even without access to a high-performance computing platform.

\begin{table}[H]
    \centering
    \begin{tabular}{||c | c ||} 
    \hline
     Learning Algorithm & Time (hours) \\ [0.5ex] 
     \hline
     GLM & 2.2 \\ 
     \hline
     GLMNet & 4.5 \\
     \hline
     XGBoost & 8.8 \\
     \hline
     SL & 30 \\
    \hline
    \end{tabular}
    \captionsetup{justification=centering}
    \caption{PheWAS run times for various nuisance parameters estimation strategies.}
    \label{tab:phewas}
\end{table}

We now turn to the GWAS setting for which run time estimates are provided in Table~\ref{tab:gwas}.
In this case, $\hat{Q}$ and $\hat{g}$ need to be estimated for each variant.
In order to obtain a run time estimate for a GWAS it is sufficient to compute the run time for one variant and simply multiply by the number of genotyped variants in a typical GWAS; here we take $600\,000$.
However, because the run time of the propensity score fit varies depending on the variant, we instead run the TMLE process over 100 variants and report the mean run time as a more accurate estimate.
While it is not currently possible to run a GWAS on a personal laptop, we find that access to a modern computing platform makes this kind of study feasible using TarGene.

\begin{table}[H]
    \centering
    \begin{tabular}{||c | c | c ||} 
    \hline
     Learning Algorithm & Unit Time (seconds) & Projected GWAS Time (hours) \\ [0.5ex] 
     \hline
     GLM & 13 & 10\\ 
     \hline
     GLMNet & 57 & 48 \\
     \hline
     XGBoost & 95 & 72 \\
     \hline
     SL & 451 & 375 \\
    \hline
    \end{tabular}
    \captionsetup{justification=centering}
    \caption{GWAS run times. The unit time corresponds to a single variant/trait pair. The projected GWAS time assumes $600\,000$ variants and $200$ folds parallelization.}
    \label{tab:gwas}
\end{table}

\subsubsection{Sieve Variance Plateau Process}

We present benchmarks for Sieve Variance Plateau estimation on a 20-core compute node in Table~\ref{tab:sieve_timing}.
For each parameter, the procedure computes a given number of variance estimates, here either $10$ or $100$ $\tau$-values.
As can be seen from the table, computational time increases sub-linearly with both the number of estimates and the number of parameters.
This is mainly because, for a small number of parameters, the computational time is driven by reading the GRM from disk hence under-using the multi-threading power of the node.
As soon as the number of parameters becomes large, \eg, for a GWAS or PheWAS, all cores can be utilised simultaneously to maximum efficiency.
Given the computational complexity of the SVP correction, we recommend the process be run for estimates just below the required significance threshold only.

\begin{table}[h!]
    \centering
    \begin{tabular}{||c | c | c | c||} 
    \hline
     Number of estimates per curve & Number of parameters & Time (s) & Time (h) \\ [0.5ex] 
     \hline
     100 & 419 & 71836 & 19.9 \\ 
     \hline
     10 & 419 & 11926 & 3.3 \\ 
     \hline
     100 & 1 & 30966 & 8.6 \\
    \hline
    10 & 1 & 6779 & 1.8 \\
    \hline
    \end{tabular}
    \captionsetup{justification=centering}
    \caption{Sieve Plateau variance estimation benchmarks.}
    \label{tab:sieve_timing}
\end{table}

\subsection{Supported databases and workflows}
The TarGene \texttt{Nextflow} pipeline can be run seamlessly on two major databases: (i) the UK Biobank (downloaded data)~\citep{UKB_bycroft}, and (ii) the cloud-based All of Us Researcher Workbench~\citep{AllOfUs}.
Inbuilt support for future use cases, such as the upcoming cohort \href{https://ourfuturehealth.org.uk}{Our Future Health}, will be added when available.
TarGene supports multiple workflows and study designs: Genome-wide association study (GWAS), phenome-wide associated study (PheWAS), and custom study design of (i) single or joint variant effects on outcome, (ii) gene-by-gene ($G \times G$) effects up to any order, and (iii) gene-by-environment ($G \times E$) effects up to any order.

%% file: appendix.tex
\appendix

\section{}

\subsection{(Weighted) fluctuation solves the efficient influence curve}\label{app:fluctuation}
We show that applying the (weighted) TMLE update step solves the efficient influence function.
The output of either procedure is the vanishing of the first-order bias,
\begin{equation}
    B_n(Q^*_n,g_n) \equiv \frac{1}{n} \sum_{i=1}^{n} D^{*}\bigl(Q^*_n,g_n\bigr)(o_i) = 0,
\end{equation}
where $D^{*}\bigl(Q^*_n,g_n\bigr)$ denotes the IF of the target parameter.
The update step can be performed in two different ways: (i) by \emph{maximising} the log-likelihood of the fluctuation $Q_{n,\epsilon}$ through the initial fit $Q_n$, or (ii) by \emph{minimising} an appropriately chosen loss function $\mathcal{L}(Q_{n,\epsilon})$ of the fluctuation $Q_{n,\epsilon}$ through the initial fit $Q_n$ as in Eqs.~\ref{eq:fluctuations},~\ref{eq:fluctuations_weighted}.
The former approach is called targeted maximum likelihood estimation whereas the latter is called targeted minimum loss-based estimation; both are abbreviated TMLE.
We apply the loss-based approach to estimate the ATE with clever covariate
\begin{equation}
    H(g_n)(A,W) = \frac{2A-1}{g_n(A,W)}
\end{equation}
or, alternatively, weighted clever covariate $H'(A,W) = 2A-1$.

\begin{prop}
    Let $Y$ be a binary outcome, let $f$ be a function of the data $O$ taking values in the unit interval $[0,1]$, and consider the log-loss function
    \begin{equation}
        \mathcal{L}[f](O) = -\log \Bigl\{ f(O)^Y [1-f(O)]^{1-Y}\Bigr\}.
    \end{equation}
    Applying TMLE by (iteratively) minimising the log-loss function $\mathcal{L}[Q^k_n(\epsilon)]$ of the fluctuation $Q^k_n(\epsilon)$ with respect to $\epsilon$ given by the logistic regression
    \begin{equation}
        Q^k_n(\epsilon) = \expit \bigl\{ \logit Q^{k}_n + \epsilon H \bigr\}, \quad \text{where} \quad k \geq 0,
    \end{equation}
    using $\hat{\epsilon} = \arg \min_{\epsilon} \sum_{i}^n \mathcal{L}\bigl[Q^k_n(\epsilon)\bigr](o_i)$ to define $Q^{k+1}_n \equiv Q^k_n(\hat{\epsilon})$, solves the empirical IF. 
    
    Similarly, let $g$ be the propensity score, and consider the $g$-weighted log-loss function
    \begin{equation}
        \mathcal{L}_g[f](O) = -\frac{1}{g(O)} \log \Bigl\{ f(O)^Y [1-f(O)]^{1-Y}\Bigr\}.
    \end{equation}
    Applying wTMLE by (iteratively) minimising the weighted log-loss function $\mathcal{L}_g[\tilde{Q}^k_n(\epsilon)]$ of the weighted fluctuation $\tilde{Q}^k_n(\epsilon)$ with respect to $\epsilon$ given by the logistic regression
    \begin{equation}
        \tilde{Q}^k_n(\epsilon) = \expit \bigl\{ \logit \tilde{Q}^k_n + \epsilon H' \bigr\}, \quad \text{where} \quad \tilde{Q}^0_n = Q^0_n, \, k \geq 0,
    \end{equation}
    using $\tilde{\epsilon} = \arg \min_{\epsilon} \sum_{i}^n\mathcal{L}_g \bigl[\tilde{Q}^k_n(\epsilon)\bigr](o_i)$ to define $\tilde{Q}^{k+1}_n \equiv \tilde{Q}^k_n(\tilde{\epsilon})$, solves the empirical IF. 
\end{prop}
\begin{proof}
    To see this, recall that $\expit(x) = \{1+\exp(-x)\}^{-1}$ is the inverse of the function $\logit(p) = \log\{p/(1-p)\}$ for $p \in (0,1)$, and for any $a,b \in \mathbb{R}$ we have the relation
    \begin{equation}\label{eq:expit}
        \frac{d}{d \epsilon}\Big|_{\epsilon=0} \expit(a+b\epsilon) 
        = b \expit(a) \bigl\{1-\expit(a)\bigr\}
    \end{equation}
    as is easily checked.
    First, we minimise the (unweighted) log-loss function:
    \begin{align}
        \frac{d}{d\epsilon}\Big|_{\epsilon = 0} \mathcal{L}\bigl[Q^k_n(\epsilon)\bigr]
        &= -Y \frac{d}{d\epsilon} \log Q^k_n(\epsilon) \Big|_{\epsilon = 0} - (1-Y) \frac{d}{d\epsilon} \log \bigl\{ 1-Q^k_n(\epsilon)\bigr\}  \Big|_{\epsilon = 0} \\
        &= -Y \frac{1}{Q^k_n} \frac{d}{d\epsilon} Q^k_n(\epsilon) \Big|_{\epsilon = 0} + (1-Y) \frac{1}{1-Q^k_n} \frac{d}{d\epsilon} Q^k_n(\epsilon) \Big|_{\epsilon = 0} \\
        &= \left\{ \frac{1-Y}{1-Q^k_n} - \frac{Y}{Q^k_n} \right\} \frac{d}{d\epsilon} Q^k_n(\epsilon) \Big|_{\epsilon = 0}.
    \end{align}
    In order to evaluate this expression, we use the result of Eq.~\ref{eq:expit} to compute
    \begin{equation}
        \frac{d}{d \epsilon} Q^k_n(\epsilon) \Big|_{\epsilon=0} = H Q^k_n \bigl(1- Q^k_n\bigr).
    \end{equation}
    Putting both computations together, we conclude
    \begin{equation}
        \frac{d}{d\epsilon}\Big|_{\epsilon = 0} \mathcal{L}\bigl[Q^k_n(\epsilon)\bigr] = H \Bigl\{ (1-Y)Q^k_n - Y\bigl(1-Q^k_n\bigr) \Bigr\} = - H\bigl(Y - Q^k_n\bigr)
    \end{equation}
    which is the first component of the IF of the ATE.
    Thus if $Q^k_n$ has been updated to the final $Q^*_n$ such that the derivative with respect to $\epsilon$ of the empirical mean of the log-loss function applied to a further fluctuation $Q^*_n(\epsilon)$ vanishes, we have
    \begin{equation}
        0 = \frac{d}{d\epsilon}\Big|_{\epsilon = 0} \sum_{i=1}^n \mathcal{L}\bigl[Q^*_n(\epsilon)\bigr](o_i) = - \sum_{i=1}^n H(o_i) \bigl\{ y_i - Q^*_n(o_i) \bigr\}.
    \end{equation}
    This means that the empirical mean of the IF vanishes, as required.
    
    Next, using similar arguments, we minimise the $g$-weighted log-loss function:
    \begin{equation}
        \frac{d}{d\epsilon}\Big|_{\epsilon = 0} \mathcal{L}_g\bigl[\tilde{Q}^k_n(\epsilon)\bigr] = \frac{1}{g} \left\{ \frac{1-Y}{1-\tilde{Q}^k_n} - \frac{Y}{\tilde{Q}^k_n} \right\} \frac{d}{d\epsilon} \tilde{Q}^k_n(\epsilon) \Big|_{\epsilon = 0}.
    \end{equation}
    In order to evaluate this expression, we again use the result of Eq.~\ref{eq:expit} to compute
    \begin{equation}
        \frac{d}{d \epsilon} \tilde{Q}^k_n(\epsilon) \Big|_{\epsilon=0} = H' \tilde{Q}^k_n \bigl(1- \tilde{Q}^k_n\bigr).
    \end{equation}
    Putting both computations together, noting that $H'/g = H$, we conclude
    \begin{equation}
        \frac{d}{d\epsilon}\Big|_{\epsilon = 0} \mathcal{L}_g\bigl[\tilde{Q}^k_n(\epsilon)\bigr] = H \Bigl\{ (1-Y)\tilde{Q}^k_n - Y\bigl(1-\tilde{Q}^k_n\bigr) \Bigr\} = - H\bigl(Y - \tilde{Q}^k_n\bigr)
    \end{equation}
    and we conclude by the same argument as for the (unweighted) log-loss function.
\end{proof}
A similar argument can be made for solving the IF with (w)TMLE when dealing with ATE for a continuous outcome $Y$. 
In this case, one combines the squared-loss function $\mathcal{L}[f](O) = \bigl\{Y - f(O)\bigr\}^2$ with the fluctuation given by a linear regression $Q^0_n(\epsilon) = Q^{0}_n + \epsilon H$ where $H$ is the clever covariate as above.
For wTMLE, one combines the $g$-weighted squared-loss function
\begin{equation}
    \mathcal{L}_g[f](O) = \frac{1}{g} \bigl\{Y - f(O)\bigr\}^2
\end{equation}
with the fluctuation given by the linear regression $\tilde{Q}^0_n(\epsilon) = Q^{0}_n + \epsilon H'$.

%% file: supplementary.tex
\section{Supplementary Material}
\label{sec:supp_figures}
\beginsupplement

\subsection{PCA Analysis}

\begin{figure}[H]
    \centering
    \includegraphics[scale=0.58]{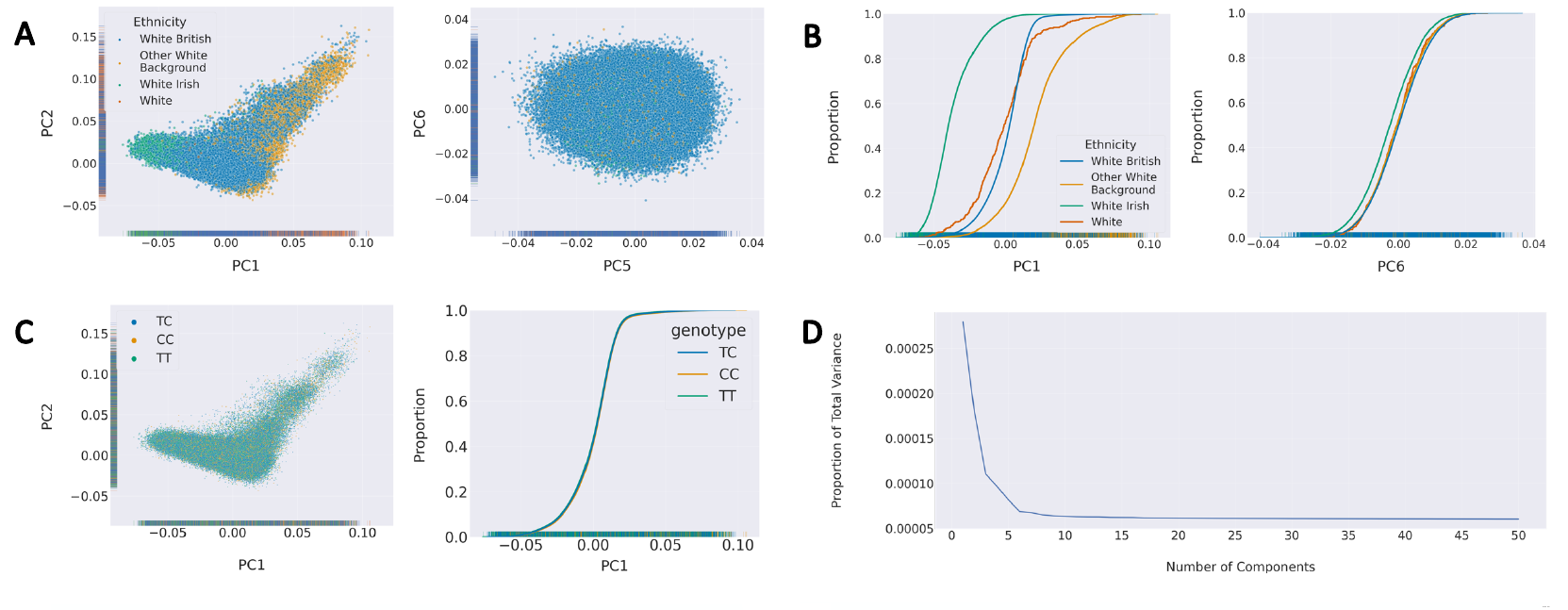}
    \caption{\textbf{Principal Component Analysis of the UK Biobank's white population} ({\bf A}) Principal component analysis labelled by ethnicity. Left: PC1 vs PC2 shows high level of population structure dependent on self-reported ethnicity. Right: PC5 vs PC6 shows a more symmetric shape suggesting that there is no ethnicity structure for PCs $>6$. This is more clearly visible in ({\bf B}) via the cumulative distribution analysis of ethnicity for PC1 and PC6. Left: The cumulative distributions of PC1 conditioned on self-reported ethnicity differ, indicating that variation in ethnicity and variation in PC1 are dependent. Right: In PC6 this separation has disappeared.
   ({\bf C}) A variant specific analysis showing that rs1421085 is not stratified in the population. The same pattern of non-stratification holds for the first 20 principal components, see Supplementary Fig.~\ref{fig:appendix-rs1421085-UKB-pcs}. When this is the case, principal components are not confounding the genotype-phenotype relationship.
   ({\bf D}) This scree plot shows that the proportion of variance explained by each additional PC plateaus after 6 PCs, when subset on `self-reported White' UK Biobank population, indicating that 6 PCs is sufficient to explain the population structure of this cohort.}
   \label{fig:pca}
\end{figure}

\clearpage

\subsection{Supplement to section~\ref{sec:simulation-estimands}}

\begin{table}[H]
    \resizebox{\columnwidth}{!}{%
    \begin{filecontents*}{estimands_description.csv}
TYPE,OUTCOME,OUTCOMETYPE,VARIANTS,OUTCOMEFREQ,TREATMENTMINFREQ,JOINTMINFREQ
AIE,Body mass index,Continuous,rs62107261;rs9940128,,4.2e-04,
AIE,Leukocyte count,Continuous,rs3859191;rs9268219,,3.3e-03,
AIE,Multiple sclerosis,Binary,rs10419224;rs59103106,4.0e-03,3.9e-04,2.1e-06
AIE,Multiple sclerosis,Binary,rs3129889;rs62295911,4.0e-03,8.4e-05,2.1e-06
AIE,Other diseases of the digestive system,Binary,rs3129716;rs72926466,5.3e-02,3.0e-04,1.8e-05
AIE,Parkinson's disease,Binary,rs11868112;rs356219;rs6456121,6.5e-03,2.1e-03,1.8e-05
AIE,Parkinson's disease,Binary,rs1732170;rs356219;rs456998;rs8111699,6.5e-03,1.4e-03,6.2e-06
AIE,Psoriasis,Binary,rs10132320;rs974766,1.0e-02,1.0e-04,4.1e-06
AIE,Sarcoidosis (D86),Binary,rs148515035;rs502771,2.3e-03,2.5e-05,2.1e-06
AIE,Skin colour,Count,rs1129038;rs1805008,,3.6e-04,
AIE,Skin colour,Count,rs1805005;rs6059655,,1.6e-04,
AIE,Skin colour,Count,rs1805007;rs6088372,,3.0e-04,
AIE,Type 2 diabetes,Binary,rs117737810;rs4506565,8.8e-03,7.8e-05,2.1e-06
AIE,psoriasis,Binary,rs10132320;rs974766,1.2e-02,1.0e-04,2.1e-06
AIE,sarcoidosis,Binary,rs148515035;rs502771,2.1e-03,2.5e-05,2.1e-06
ATE,Body mass index,Continuous,rs62107261,,2.2e-03,
ATE,Body mass index,Continuous,rs9940128,,1.8e-01,
ATE,Leukocyte count,Continuous,rs3859191,,2.2e-01,
ATE,Leukocyte count,Continuous,rs9268219,,1.5e-02,
ATE,Multiple sclerosis,Binary,rs3129889,4.0e-03,2.0e-02,2.5e-04
ATE,Multiple sclerosis,Binary,rs62295911,4.0e-03,3.1e-03,2.7e-05
ATE,Other diseases of the digestive system,Binary,rs3129716,5.3e-02,2.1e-02,1.9e-03
ATE,Other diseases of the digestive system,Binary,rs72926466,5.3e-02,1.6e-02,8.4e-04
ATE,Sarcoidosis (D86),Binary,rs148515035,2.3e-03,3.7e-04,2.1e-06
ATE,Sarcoidosis (D86),Binary,rs502771,2.3e-03,7.5e-02,3.1e-04
ATE,Type 2 diabetes,Binary,rs117737810,8.8e-03,9.2e-04,1.0e-05
ATE,Type 2 diabetes,Binary,rs4506565,8.8e-03,1.0e-01,1.4e-03
ATE,sarcoidosis,Binary,rs148515035,2.1e-03,3.7e-04,4.1e-06
ATE,sarcoidosis,Binary,rs502771,2.1e-03,7.5e-02,3.2e-04
\end{filecontents*}
\csvreader[
        tabular = |c|c|c|c|c|c|c|,
        table head = \hline Type & Outcome & \multicolumn{1}{|p{2cm}|}{\centering Outcome \\ Type} & Variants & \multicolumn{1}{|p{2cm}|}{\centering Outcome \\ Freq} & \multicolumn{1}{|p{2cm}|}{\centering Variants \\ Min \\ Freq} & \multicolumn{1}{|p{2cm}|}{\centering Joint \\ Min \\ Freq} \\\hline\hline,
        table foot = \hline
    ]{estimands_description.csv}{
    TYPE=\TYPE,OUTCOME=\OUTCOME,OUTCOMETYPE=\OUTCOMETYPE,VARIANTS=\VARIANTS,OUTCOMEFREQ=\OUTCOMEFREQ,TREATMENTMINFREQ=\TREATMENTMINFREQ,JOINTMINFREQ=\JOINTMINFREQ}{
    \TYPE & \OUTCOME & \OUTCOMETYPE & \VARIANTS & \OUTCOMEFREQ & \TREATMENTMINFREQ & \JOINTMINFREQ}
    }
    \caption{The $29$ estimands used across the simulation study. The ``Variants Min Freq" column represents the minor genotype frequency for the variants in the estimand. When the outcome is binary, the frequency is provided as well as the ``Joint Min Freq". The latter represents the minor frequency of joint (genotype, outcome).}
    \label{table:simulation_estimands}
\end{table}

\subsection{Supplement to section~\ref{sec:realistic-simulation}}

\begin{algorithm}[H]
    \caption{\textbf{Sieve Neural Network Estimator.} The undefined ``train" function corresponds to each neural network's training loop and implicitly uses early-stopping to control the number of training epochs of the proposed architecture.}
    \label{alg:sieve_network}
    \begin{algorithmic}
    \Procedure{SNNE}{$candidateModels, dataset, maxSievePatience$}
    
    \State $trainingSet, validationSet \gets split(dataset)$
    \State $bestModel \gets candidateModels[1]$
    \State $bestValidationLoss \gets train(bestModel, trainingSet, validationSet)$
    \State $sievePatience \gets 0$
    
    \For{$model \in candidateModels[2:end]$}
        \State $validationLoss \gets train(model, trainingSet, validationSet)$
        
        \If{$validationLoss \leq bestValidationLoss$}
            \State $bestValidationLoss \gets validationLoss$
            \State $bestModel \gets model$
            \State $sievePatience \gets 0$
        \Else
            \State $sievePatience \gets sievePatience + 1$
        \EndIf

        \If{$sievePatience == maxSievePatience$}
            \State \textbf{break}
        \EndIf
    \EndFor
        \State \textbf{return} $bestModel$
    \EndProcedure
    
    \end{algorithmic}

\end{algorithm}

\begin{figure}[H]
    \centering
    \includegraphics[scale=0.2]{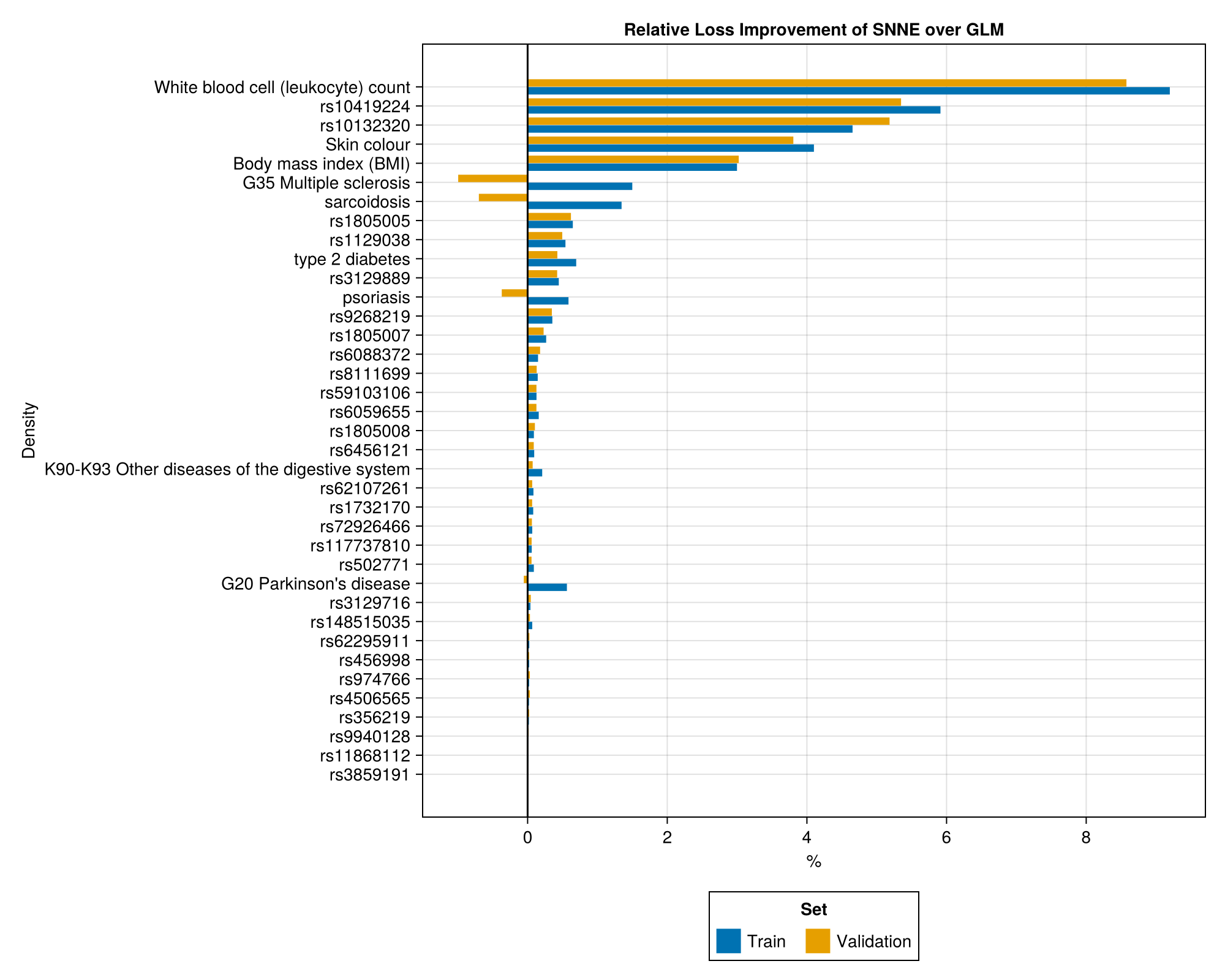}
    \caption{Comparison of the empirical loss between the proposed Sieve Neural Network Estimator and a Generalised Linear Model baseline. For each density (y-axis), results are presented as a relative improvement of the SNNE over the GLM (x-axis). Bars facing to the right of the thick 0-line indicate an improvement while bars facing to the left indicate a deterioration of the loss. Both Train (Blue) and Validation (Yellow) set improvements are presented. These results validate the proposed density estimation strategy as an effective flexible and data-adaptive method.}
    \label{fig:density-estimates-comparison}
\end{figure}

\begin{figure}
    \centering
    \includegraphics[scale=0.3]{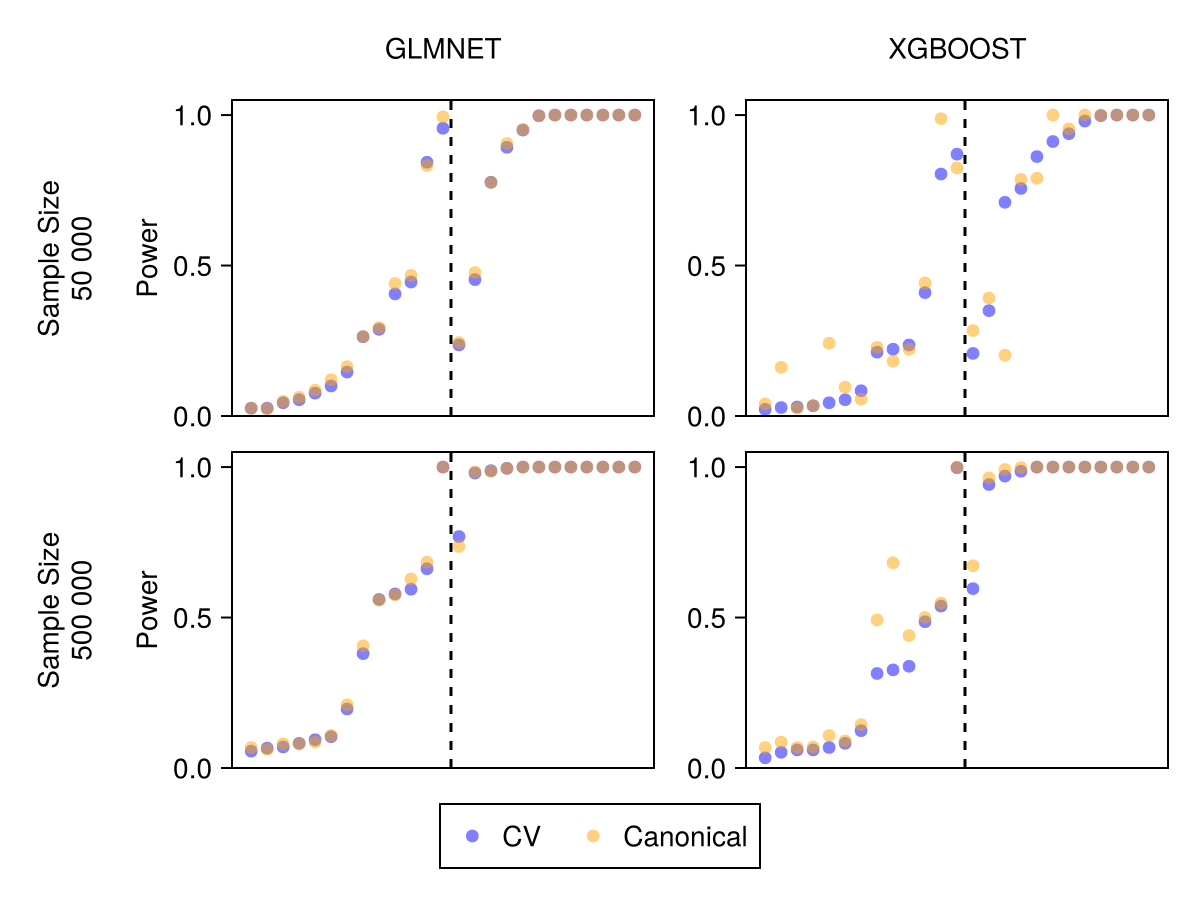}
    \caption{\textbf{Power analysis of OSE.} The plot is organised in 4 quadrants. Rows indicate sample sizes with $n = 50\, 000$ (top) and $n = 500\, 000$ (bottom), columns indicate the model used to fit nuisance functions with GLMNet (left) and XGBoost (right), and colour indicates the resampling scheme, \ie, cross-validated (blue) and canonical (orange). Each dot corresponds to a single estimand and the dashed lines separate AIEs (left) and ATEs (right).}
    \label{fig:power-ose}
\end{figure}

\clearpage

\subsection{Supplement to Fig.~\ref{fig:rs1421085-plots}: PCs for rs1421085}

\begin{figure}
    \centering
    \includegraphics[width=\textwidth]{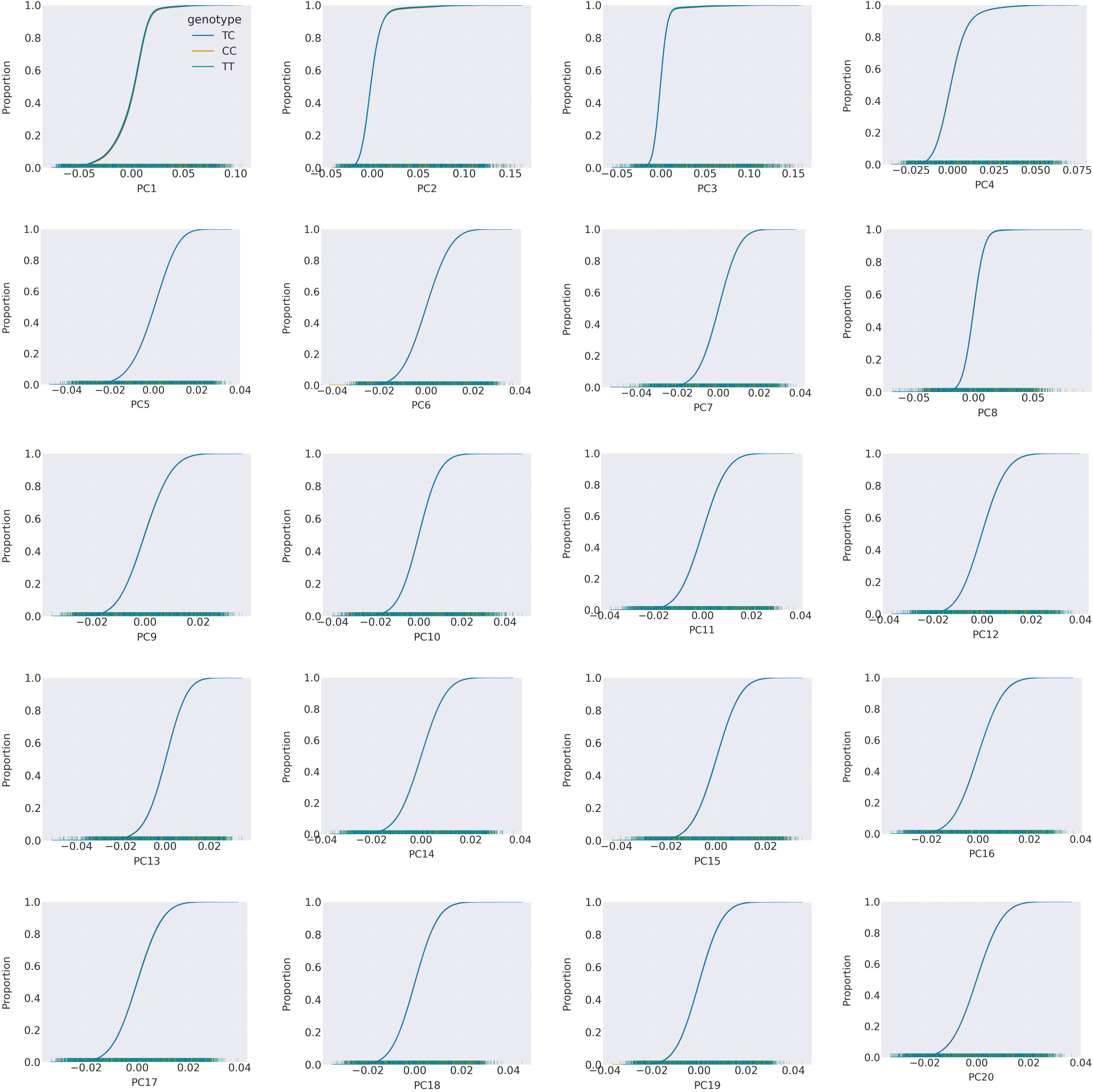}
    \caption{\textbf{Cumulative density functions of the first 20 principal components stratified by rs1421085 genotype.} Population stratification plays no discernible role in the genotype distribution of the FTO variant rs1421085 in the UKB population with white ethnic background.}
    \label{fig:appendix-rs1421085-UKB-pcs}
\end{figure}

\subsection{Supplement to Fig.~\ref{fig:rs1421085-plots}: TarGene results}
Fig.~\ref{fig:appendix-rs1421085-plots}A provides further information on the Sieve Plateau variance corrected p-values, whereas Fig.~\ref{fig:appendix-rs1421085-plots}B shows an example Sieve Variance Plateau curve.
Fig.~\ref{fig:appendix-rs1421085-plots}C illustrates the difference between the initial estimate, reported by Super Learning, and TMLE after the targeting step.
Fig.~\ref{fig:appendix-rs1421085-plots}D provides a histogram of the genetic relationship matrix values on the UK Biobank population.

\begin{figure}[H]
    \centering
    \includegraphics[width=\textwidth]{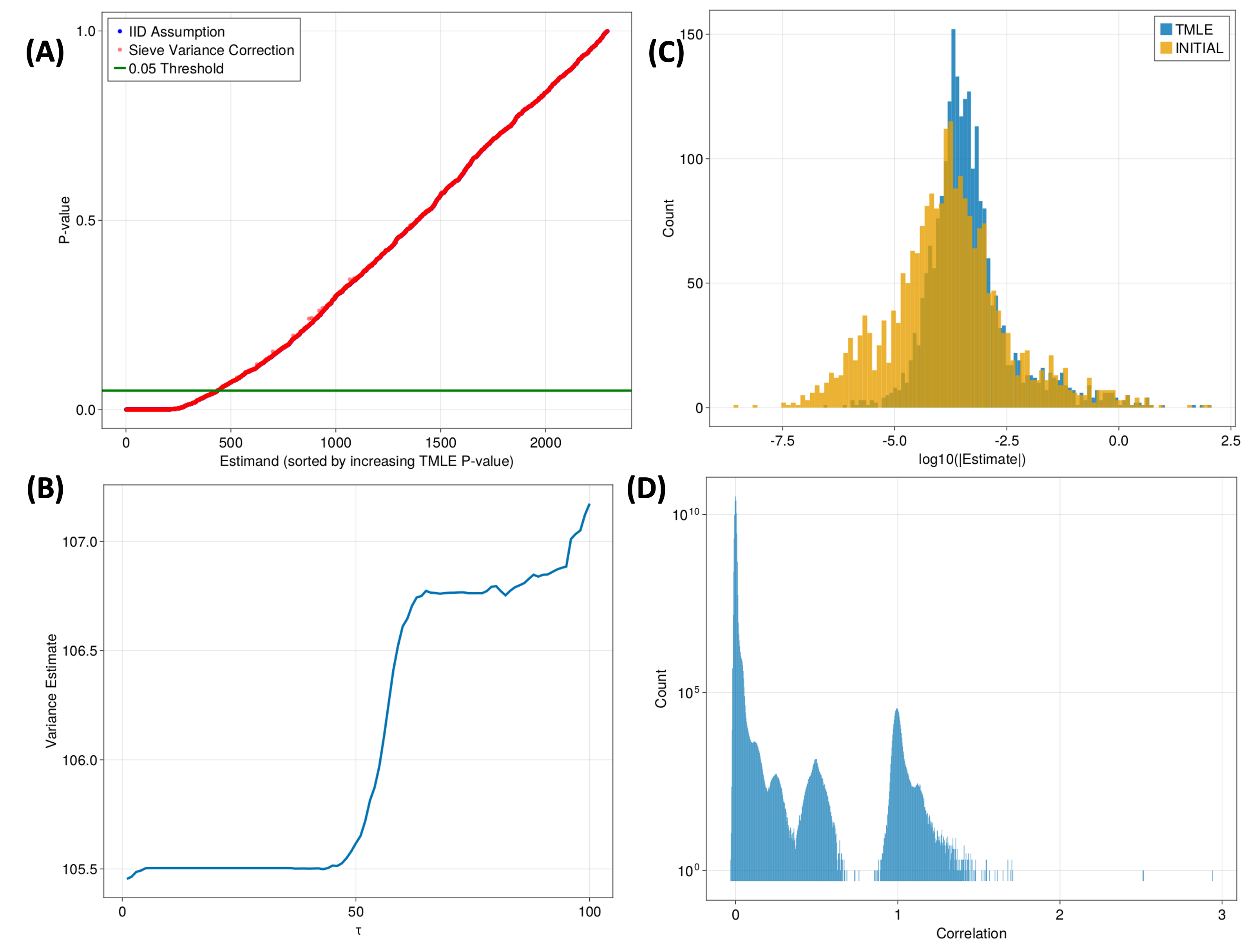}
    %\captionsetup{width=\textwidth}
    \caption{\textbf{Supplementary figures related to Targeted Maximum Likelihood Estimation and Sieve Variance Plateau correction.}\\
    ({\bf A})~\textbf{Sieve Variance Plateau corrected p-values.} Evidence of no significant difference across traits for rs1421085 between corrected p-values and p-values assuming individuals are independent.
    ({\bf B})~\textbf{Sieve Variance Plateau curve.} A sample Sieve-Variance-Plateau curve for body mass index across 100 different thresholds.
    ({\bf C}) \textbf{Difference between initial and TMLE estimates.} The initial estimate is obtained by plugin of our first Super Learning estimate for $Q$. In more than $78\%$ of the cases, TMLE is driving the initial estimate towards more extreme values. This is a potential piece of evidence for the omnigenic model~\citep{Omnigenic_model}.
    ({\bf D}) \textbf{Distribution of the Genetic Relationship Matrix.} Most individuals in the UK-Biobank have low genetic similarity.
    } 
    \label{fig:appendix-rs1421085-plots}
\end{figure}

\begin{figure}[H]
    \centering
    \includegraphics[scale=1.0]{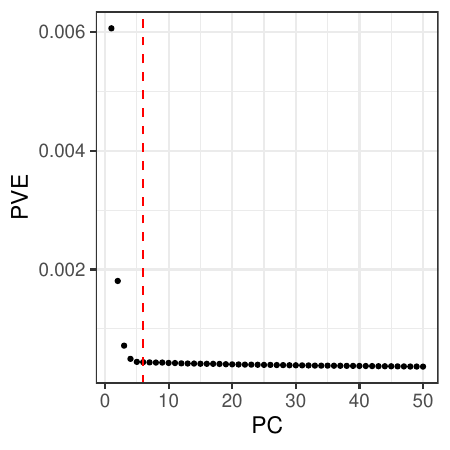}
    %\captionsetup{width=\textwidth}
    \caption{\textbf{All of Us (AoU) Principal Component Analysis.}
    Principal components were computed using FlashPCA2 across all genotyped SNPs, excluding SNPs in linkage disequilibrium with rs1421085, for 122,752 participants in the All of Us cohort. The y-axis shows the proportion of variance explained (y-axis; PVE) across the first 50 principal components (x-axis; PCs), with a dashed red line at 6 PCs. The proportion of variance explained plateaus at 6 PCs, and is sufficient to capture confounding due to population stratification in the All of Us cohort study described in this article. This is consistent with results found in the UKB study for the same variant.
    } 
    \label{fig:appendix-rs1421085-AoU-screeplot}
\end{figure}

\begin{figure}[H]
    \centering
    \includegraphics[scale=0.8]{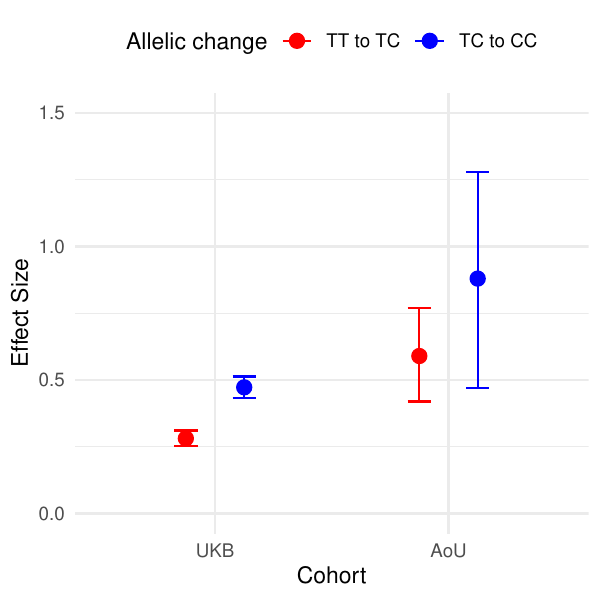}
    \includegraphics[scale=0.63]{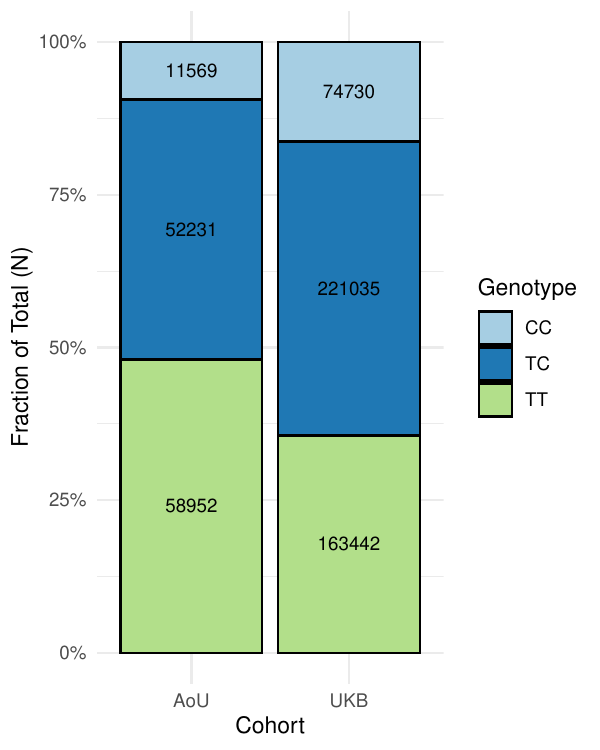}
    %\captionsetup{width=\textwidth}
    \caption{\textbf{All of Us (AoU) cohort compared to UKB for the effect of rs1421085 on BMI.}
    \textbf{(Left)} Estimates of the individual allelic effects of rs1421085 on BMI for the UKB and AoU cohort with 95\% confidence intervals. Confidence intervals do not overlap for TT-to-TC effect estimates (red), but they do overlap for TC-to-TT estimates (blue) across cohorts. A  higher degree of uncertainty is found in estimates for the AoU cohort. Although a significant non-linear effect was found in the UKB for this association, this is not reproduced in the AoU cohort. This may be due to reduced power as a result of AoU's lower sample size, or due to increased variation as a result of larger environmental effects. \textbf{(Right)} Decrease in minor allele frequency ($0.40$ in UKB; $0.30$ in AoU) and sample size ($459,\!207$ in UKB; $122,\!752$ in AoU). Here we show a breakdown of sample sizes for each genotype across the UKB and AoU cohort, with the fraction of the cohort shown on the y-axis and the sample size for each group labelled in text.
    } 
    \label{fig:appendix-rs1421085-BMI-AoU-UKB-effects}
\end{figure}

\clearpage

\subsection{Supplement to Section~\ref{sec:simulation_results}}.

\subsubsection{Definition of Bootstrap Bias, Variance and MSE estimators.} \label{sec:bootstrap_estimators}

Let $n$ be the number of independent and identically distributed samples ($n=50\,000$ or $500\,000$) and let $B$ be the number of bootstrap resamples ($B=500$).
For a given estimand and estimator, each bootstrap resample yields a $p$-dimensional vector of estimates denoted by $\hat{\mathbf{\Psi}}_{b,n}$.
For this estimate, a true value $\mathbf{\Psi}_0$ is available.
We denote by $\mathbf{\Sigma}_{B, n}$ the $p \times p$-dimensional sample covariance matrix of the estimator $\hat{\mathbf{\Psi}}_{n}$ which is obtained from the $B$ estimate vectors $\hat{\mathbf{\Psi}}_{b,n}$ for $b = 1, 2, \ldots, B$. 
We then estimate the bootstrap bias-squared, the variance and mean-squared error of our high-dimensional estimators via the following formulas:
\begin{align} 
    \widehat{\Bias}^2_{B, n} &= \frac{1}{B}\sum_{b=1}^B ||\hat{\mathbf{\Psi}}_{b,n} - \mathbf{\Psi}_0||_2^2 \label{eq:bootstrap-bias}\\
    \widehat{\Var}_{B, n} &= \Tr(\mathbf{\Sigma}_{B, n}) \label{eq:bootstrap-var}\\
    \widehat{\MSE}_{B, n} &= \widehat{\Bias}^2_{B, n} + \widehat{\Var}_{B, n} \label{eq:bootstrap-mse}
\end{align}
Here $\Tr$ denotes the trace of a matrix.
In particular, the bootstrap variance $\widehat{\Var}_{B,n}$ is the sum of the diagonal elements of the sample covariance matrix $\mathbf{\Sigma}_{B,n}$, \ie, the sum of the component variance estimates.

\subsubsection{Comparison of influence curve-based and bootstrap variance estimates.}

We provide a comparison between the variance estimates obtained from the influence curve as per Eqs.~\ref{eq:variance_iid} and~\ref{eq:cv_variance_iid}, and the bootstrap variance from resampling as per Eq.~\ref{eq:bootstrap-var}.

\begin{figure}[h!]
    \centering
    \includegraphics[scale=0.5]{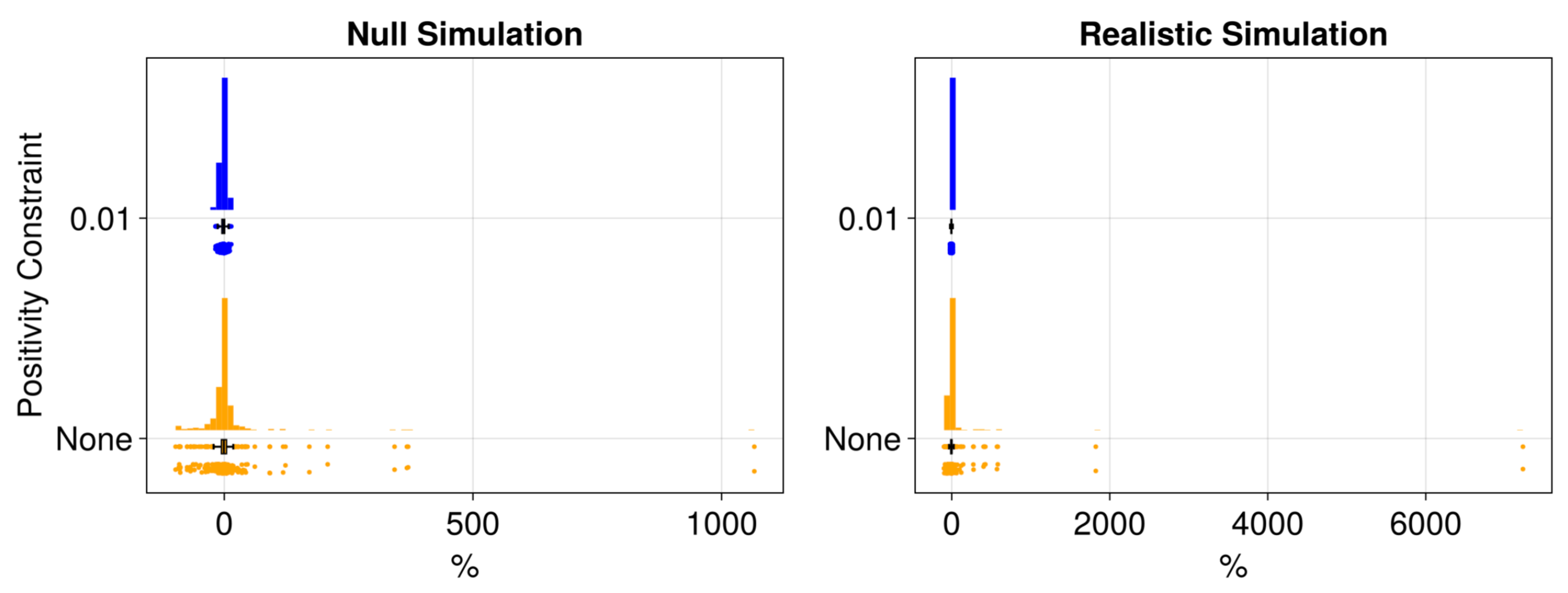}
    \caption{Relative difference in percentages between variance estimates based on influence curve and bootstrap (Eq.~\ref{eq:bootstrap-var}). When no positivity constraint is applied, variance estimates based on the influence curve may become extreme. When estimands are constrained at the $0.01$ positivity threshold level, variance estimates get closer to each other.}
    \label{fig:appendix-bootstrap-vs-infcurve-var}
\end{figure}